\documentclass[lettersize,journal,article]{IEEEtran}
\usepackage{amsmath,amsfonts}
\usepackage{algorithm, algorithmicx, algpseudocode}

\usepackage{makecell}
\usepackage{array}
\usepackage{booktabs}
\usepackage{amssymb}
\usepackage{multirow}
\usepackage[caption=false,font=normalsize,labelfont=sf,textfont=sf]{subfig}
\usepackage{caption}
\usepackage{textcomp}
\usepackage{tabularx}
\usepackage{stfloats}
\usepackage{graphicx}
\usepackage{epstopdf}
\usepackage{url}
\usepackage{color}
\usepackage{cite}
\usepackage{verbatim}
\usepackage{hyperref}
\usepackage{float}
\usepackage{amsthm}
\newtheorem{theorem}{Theorem}
\newtheorem{lemma}{Lemma}
\newtheorem{remark}{Remark}

\newtheorem{definition}{Definition}

\def\BibTeX{{\rm B\kern-.05em{\sc i\kern-.025em b}\kern-.08em
    T\kern-.1667em\lower.7ex\hbox{E}\kern-.125emX}}
\usepackage{balance}

\begin{document}
\title{Parallel Collaborative ADMM Privacy Computing and Adaptive GPU Acceleration for Distributed Edge Networks}

\author{Mengchun Xia, Zhicheng Dong,~\IEEEmembership{Member,~IEEE}, Donghong Cai,~\IEEEmembership{Senior Member,~IEEE}, Fang Fang,~\IEEEmembership{Senior Member,~IEEE}, Lisheng Fan,~\IEEEmembership{Member,~IEEE}, and Pingzhi Fan,~\IEEEmembership{Life Fellow,~IEEE}

\thanks{The work of Zhicheng Dong was supported by the Science and Technology Major Project of Tibetan Autonomous Region of China (NO. XZ202502ZY0070), the Science and Technology Major Project of Tibetan Autonomous Region of China (No. XZ202201ZD0006G); the work of Donghong Cai was supported by the Basic and Applied Basic Research Foundation of Guangdong province (No. 2024A1515012398); the work of Lisheng Fan was supported by the Science and Technology Program of Guangzhou (No. SL2023A03J00284). (Corresponding author: Zhicheng Dong.)

M. Xia and Z. Dong are with the College of Information Science and Technology, Tibet University, Lhasa, 850000, China (e-mail: mengchunxiavip@163.com; dongzc666@163.com).

D. Cai is with the College of Cyber Security, Jinan University, Guangzhou 510632, China (e-mail:  dhcai@jnu.edu.cn).

F. Fang is with the Department of Electrical and Computer Engineering, Western University, London N6A 5B9, Canada (e-mail: fang.fang@uwo.ca).

L. Fan is with the School of Computer Science, Guangzhou University, Guangzhou 510006, China (e-mail: lsfan@gzhu.edu.cn).

P. Fan is with the Key Lab of Information Coding and Transmission, Southwest Jiaotong University, Chengdu 610031, China (e-mail: pzfan@swjtu.edu.cn).}}

\markboth{Journal of \LaTeX\ Class Files,~Vol.~, No.~, ~2025}%
{How to Use the IEEEtran \LaTeX \ Templates}

\maketitle

\begin{abstract}
Distributed computing has been widely applied in distributed edge networks for reducing the processing burden of high-dimensional data centralization, where a high-dimensional computational task is decomposed into multiple low-dimensional collaborative processing tasks or multiple edge nodes use distributed data to train a global model. However, the computing power of a single-edge node is limited, and collaborative computing will cause information leakage and excessive communication overhead. In this paper, we design a parallel collaborative distributed alternating direction method of multipliers (ADMM) and propose a three-phase parallel collaborative ADMM privacy computing (3P-ADMM-PC2) algorithm for distributed computing in edge networks, where the Paillier homomorphic encryption is utilized to protect data privacy during interactions. Especially, a quantization method is introduced, which maps the real numbers to a positive integer interval without affecting the homomorphic operations. To address the architectural mismatch between large-integer and Graphics Processing Unit (GPU) computing, we transform high-bitwidth computations into low-bitwidth matrix and vector operations. Thus the GPU can be utilized to implement parallel encryption and decryption computations with long keys. Finally, a GPU-accelerated 3P-ADMM-PC2 is proposed to optimize the collaborative computing tasks. Meanwhile, large-scale computational tasks are conducted in network topologies with varying numbers of edge nodes. Experimental results demonstrate that the proposed 3P-ADMM-PC2 has excellent mean square error performance, which is close to that of distributed ADMM without privacy-preserving. Compared to centralized ADMM and distributed ADMM implemented with Central Processing Unit (CPU) computation, the proposed scheme demonstrates a significant speedup ratio.
\end{abstract}


\section{Introduction}
\IEEEPARstart{W}{ith} the large-scale application of machine-type communication system, massive data needs to be processed effectively. Moreover, high-resolution images are widely used in medical imaging, remote sensing, intelligent driving, unmanned aerial vehicle, etc, due to the rapid development of digital image techniques. Large amounts of data are used for centralized or distributed model training to effectively serve intelligent applications. In particular, some distributed learning and computing frameworks have been proposed, which can effectively utilize distributed data without compromising the privacy. For example, federated learning \cite{10599518, wu2024client}, secure multiparty computation \cite{li2024publicly} and edge or fog computing \cite{hua2023edge, wang2023wireless}. However, distributed learning or computing \cite{liu2023asynchronous} requires strong single-node computing power and additional parameter exchange overhead. Fortunately, most of the local data is image and video, which is sparse data \cite{wang2023image}. Local learning models also need to be sparse before transmission and processing \cite{tan2023threshold} to save transmission bandwidth. Besides, large-scale random access to network nodes leads to the sparsity of distributed network topology. Therefore, distributed computing combining sparse signal processing can be widely applied to solve local high-dimensional computation problems with sparsity prior information.

The Least Absolute Shrinkage and Selection Operator (LASSO) problem is generally applied to formulate the sparse signal recovery due to the sparse solution.
Greedy algorithm, such as orthogonal
matching pursuit (OMP) \cite{wimalajeewa2014omp}, can recover sparse signals well when the compressed-sensing restricted isometry
property (RIP) \cite{huang2024sparse} is satisfied, but the complexity of OMP-like algorithms is higher. In addition, the Bayesian estimation algorithm can be well designed by using sparse prior information. For the processing of massive distributed data, \cite{wimalajeewa2014omp} proposes a large number of distributed collaborative computing algorithms to approximate the performance of the central computing algorithm, without the need for a large amount of data transmission overhead and the disclosure of local data privacy. Specially, the ADMM as an effective distributed optimization method, has been widely applied in areas such as multi-robot collaboration \cite{halsted2021survey}, wireless communication control \cite{shen2012distributed, liu2021admm}, autonomous vehicle routing \cite{zhang2021semi}, image processing \cite{afonso2010augmented, almeida2013deconvolving}, and smart grid operations \cite{liu2023distributed, yang2022optimal, long2021efficient}. The ADMM decomposes the LASSO problem into several more tractable subproblems, allowing these to be executed in parallel across multiple edge devices, and ultimately achieves consensus by iteratively updating the global solution. ADMM is particularly effective for large-scale datasets with structured sparsity in the models.

\begin{table*}[htbp]
\centering
\caption{Existing Distributed Privacy Computing Schemes}
\label{Scheme-Comparison}
\begin{tabularx}{7in}{>{\centering\arraybackslash}m{0.9in}>{\centering\arraybackslash}m{1in}>{\centering\arraybackslash}m{0.4in}>{\centering\arraybackslash}m{0.4in}>{\centering\arraybackslash}m{0.75in}>{\centering\arraybackslash}m{0.75in}>{\centering\arraybackslash}m{0.8in}>{\centering\arraybackslash}m{0.75in}}
\toprule
        \multicolumn{1}{c|}{\multirow{2}{0.85in}{\centering\textbf{Computing Algorithm}}}
      & \multicolumn{1}{c|}{\multirow{2}{1in}{\centering\textbf{Existing Schemes}}}
      & \multicolumn{2}{c|}{\centering\textbf{Privacy Computing}}
      & \multicolumn{2}{c|}{\centering\textbf{Available GPU-Acceleration}}
      & \multicolumn{2}{>{\centering\arraybackslash}m{1.55in}}{\textbf{Available Security}} \\
\cmidrule(lr){3-4}\cmidrule(lr){5-6}\cmidrule(lr){7-8}
      \multicolumn{1}{c|}{} & \multicolumn{1}{c|}{} & DP & \multicolumn{1}{c|}{HE} & Matrix & \multicolumn{1}{c|}{ModExp} & Global privacy & Local privacy \\
\midrule
      \multirow{2}{0.85in}{\centering{Dis.-OMP}} & SPriFed-OMP \cite{lin2024unified} & \checkmark & & \checkmark & & \checkmark & \\

       & CM-Pai.-OMP \cite{wang2022qgtc} & & \checkmark & \checkmark & \checkmark & & \checkmark \\
\midrule
      \multirow{2}{0.85in}{\centering{Dis.-ADMM}} & DP-ADMM \cite{huang2019dp} & \checkmark & & \checkmark & & \checkmark & \\

      & scheme from \cite{zhang2018admm} & & \checkmark & \checkmark & \checkmark & \checkmark & \\
\midrule
      \multirow{2}{0.85in}{\centering{FL}} & NbAFL \cite{wei2020federated} & \checkmark & & \checkmark & & & \checkmark \\

      & scheme from \cite{yang2024privacy} & & \checkmark & \checkmark & & \checkmark & \checkmark \\
\midrule
      \multirow{2}{0.85in}{\centering{MPL}} & PEA \cite{ruan2023private} & \checkmark & & \checkmark & & & \checkmark \\

      & PE-MPDL \cite{gong2020privacy} & & \checkmark & \checkmark & \checkmark & & \checkmark \\
\midrule
      \multirow{2}{0.85in}{\centering{Fog-Com.}} & MaxDiff \cite{piao2019privacy} & \checkmark & & \checkmark & & & \checkmark \\

      & EEDHSC \cite{gupta2021energy} & & \checkmark & \checkmark & & & \checkmark \\
\midrule
      \textbf{Ours scheme} & 3P-ADMM-PC2 & & \checkmark & \checkmark & \checkmark & \checkmark & \checkmark \\
\bottomrule
\end{tabularx}
\end{table*}

In distributed computing, such as distributed OMP (Dis.-OMP), distributed ADMM (Dis.-ADMM), federated learning (FL), multi-party learning (MPL), and fog computing (Fog-Com.), multiple edge nodes collaborate to perform computations, significantly reducing computational overhead, especially for large-scale datasets. However, as shown in Table \ref{Scheme-Comparison}, these schemes initially paid less attention to data privacy protection during node interactions. Although federated learning and multi-party learning consider the privacy protection of local data, there is still a risk of privacy leakage.

In distributed computing employing differential privacy as a privacy-preserving mechanism, such as in SPriFed-OMP \cite{lin2024unified}, DP-ADMM \cite{huang2019dp}, NbAFL \cite{wei2020federated}, PEA \cite{ruan2023private}, and MaxDiff \cite{piao2019privacy}, the addition of noise may compromise the sparsity of the original signal. Moreover, introducing noise in each gradient update can adversely affect convergence and the accuracy of the final solution. In multi-party settings, noise parameters must be meticulously designed. Therefore, in scenarios demanding high-precision privacy protection, the use of homomorphic encryption schemes becomes indispensable. In distributed computing utilizing homomorphic encryption for privacy protection, such as in CM-Paillier-OMP \cite{wang2022qgtc}, schemes from \cite{zhang2018admm}, PE-MPDL \cite{gong2020privacy}, the substantial computational overhead induced by encryption results in inefficiency for large-scale matrix operations. Such methods are often limited to small-scale experiments due to the lack of effective acceleration schemes. The work in \cite{10589503} summarizes and analyzes the advantages and limitations of homomorphic encryption as privacy-preserving solutions in distributed computing in recent years, as well as related optimization problems. It also provides new research directions for homomorphic encryption applications in distributed computing. Furthermore, deployment on resource-constrained edge devices remains particularly challenging. The work in \cite{yang2024privacy}\cite{jiang2024lancelot} combines federated learning with the CKKS scheme to defend against Byzantine attacks, inference attacks and reconstruction attacks. However, CKKS \cite{cheon2017homomorphic} is an approximate decryption scheme with complex implementation, requiring relinearization or bootstrapping operations to reduce the impact of noise during homomorphic operations. The EEDHSC in \cite{gupta2021energy} combines fog computing with the ECC-ElGamal encryption scheme. Although the ECC-ElGamal scheme \cite{elgamal1985public} provides equivalent security to RSA encryption scheme with shorter key lengths, it supports only homomorphic multiplication operations. Additionally, there are homomorphic encryption schemes such as BFV \cite{brakerski2012fully} and BGN \cite{boneh2005evaluating}. Similar to CKKS, BFV has complex implementation requiring relinearization or bootstrapping operations. BGN supports only once homomorphic multiplication. Besides, among the aforementioned homomorphic encryption schemes, Paillier, BGN, ECC-ElGamal, and BFV can only encrypt unsigned integers. However, the Paillier homomorphic encryption scheme \cite{paillier1999public} is simple to implement and supports a limited number of homomorphic additions and constant multiplications, which aligns well with the computational tasks in this paper. Therefore, considering these factors, Paillier homomorphic encryption is used in this paper to protect data privacy during interactions.

In recent years, designing customized GPU computing paradigms for specific computational scenarios has become a research hotspot. The core idea involves hardware-software co-design, mapping computational patterns onto GPU parallel architectures to unleash extreme performance. For instance, Yu et al. proposed the TwinPilots paradigm \cite{yu2024twinpilots}. It dynamically schedules Transformer layer computations across GPU and CPU cores. An online load-balancing planner coordinates GPU data loading time with CPU computation time. This achieves efficient large language model inference on memory-constrained GPU servers. Wang et al. proposed the Q-GTC paradigm \cite{wang2022qgtc}, which achieves a several-fold speedup compared to conventional CUDA cores by mapping quantized graph neural network operations onto GPU Tensor Cores to leverage their capability in executing matrix operations at specific bit-widths. Shivdikar et al. accelerated fully homomorphic encryption (FHE) \cite{shivdikar2023gme} through GPU-based microarchitectural extensions. By introducing a CU-side hierarchical interconnect in the AMD CDNA GPU architecture, they managed to retain FHE ciphertexts in cache, eliminating redundant memory accesses and reducing latency. Additionally, they designed a MOD unit and 64-bit integer arithmetic cores to provide native modular arithmetic support, accelerating the most frequent modular reduction operations in FHE, while optimizing throughput via pipelining. They also proposed a Locality-Aware Block Scheduler to enhance scheduling efficiency by leveraging the temporal locality of FHE primitive blocks, further optimizing computational performance. Riazi et al. proposed the HEAX specialized architecture \cite{riazi2020heax}, which designs a GPU-like parallel architecture tailored for homomorphic encryption. Its computing units and memory hierarchy are custom-built for cryptographic primitives, enabling hardware-level acceleration of homomorphic encryption operations. However, in the particular field of homomorphic encryption, its computational characteristics exhibit a fundamental architectural mismatch with mainstream GPU paradigms, making the direct application of existing paradigms challenging.

Facing the above-mentioned problems, in this paper, we consider a distributed edge privacy-computing network and propose a 3P-ADMM-PC2 algorithm that integrates distributed parallel ADMM with the Paillier homomorphic encryption for sparse signal processing. The proposed 3P-ADMM-PC2 decomposes a high-dimensional computational problem into multiple low-dimensional subproblems for collaborative processing by multiple edge nodes, while also protecting data privacy during node interactions. To enable encryption and decryption of real numbers and achieve homomorphic operations, a quantization method is proposed, which maps real numbers to an unsigned integer range without affecting the homomorphic properties of the encryption scheme. Compared with traditional quantization methods \cite{zhang2018admm}, our method does not require additional space to represent negative numbers. However, a vector encryption is required to reduce the computing deday in bigdata. Existing homomorphic encryption schemes show lower parallelism during vector encryption and decryption due to their computational characteristics \cite{brakerski2012fully, boneh2005evaluating, paillier1999public, elgamal1985public}, which is the reason these schemes incur significant additional computational overhead. If multiple elements within a vector can be encrypted in parallel, it would significantly reduce the computational overhead. Feasible existing approach is to deploy encryption operations on GPU or FPGA \cite{che2023unifl, dong2022tegras} to achieve parallel encryption. However, according to the 754-2008 IEEE Standard for Floating-Point Arithmetic \cite{ieee2019ieee}, each GPU core supports only 32 bits or 64 bits operations. To ensure security, a key length of at least 1024-bit is required. Modular exponentiation (ModExp) of such large integers cannot be directly implemented on GPU or FPGA. A GPU-accelerated computation scheme is proposed, which decomposes computational tasks to enable multiple GPU cores to collaboratively process a single task, thereby achieving parallel operations such as encryption overall. Finally, we propose a GPU-accelerated 3P-ADMM-PC2 scheme, which enables edge nodes to collaborate with the master node for encryption and decryption computations, thereby reducing the computational burden on the master node. Summary, the contributions of this paper are as follows:
\begin{itemize}
    \item We propose a 3P-ADMM-PC2 algorithm for large-scale sparse signal processing, integrating distributed parallel ADMM with the Paillier homomorphic encryption. In particular, proposed 3P-ADMM-PC2 decomposes high-dimensional computational tasks into low-dimensional subtasks for parallel computation by multiple edge nodes, while protecting data privacy during node interactions.
    \item A quantization method and a GPU-accelerated computation scheme are proposed to address the plaintext limitations of homomorphic encryption and deploy encryption operations on GPU. To prevent overflow in the ModExp of large integers under long key length, the scheme transforms high-bitwidth computations into low-bitwidth matrix and vector operations. This approach fully utilizes the multi-core characteristics of GPU to achieve parallel privacy-preserving computations.
    \item To reduce the computational burden on the master node, we propose a GPU-accelerated 3P-ADMM-PC2 scheme. The master node and edge nodes perform ModExp in parallel on two smaller spaces, respectively. Compared with direct computation in a large space, our approach further reduces computational overhead and enables collaborative encryption and decryption between the master node and edge nodes.
\end{itemize}

The remainder of this paper is organized as follows. Section \uppercase\expandafter{\romannumeral2} introduces the system model and problem formulation. Section \uppercase\expandafter{\romannumeral3} introduces our proposed parallel collaborative ADMM privacy computing. Section \uppercase\expandafter{\romannumeral4} details our proposed adaptive GPU accelerated 3P-ADMM-PC2. Section \uppercase\expandafter{\romannumeral5} presents experimental results and Section \uppercase\expandafter{\romannumeral6} summarizes our work.
\begin{figure*}[tp]
\centering
\includegraphics[width=7.15in]{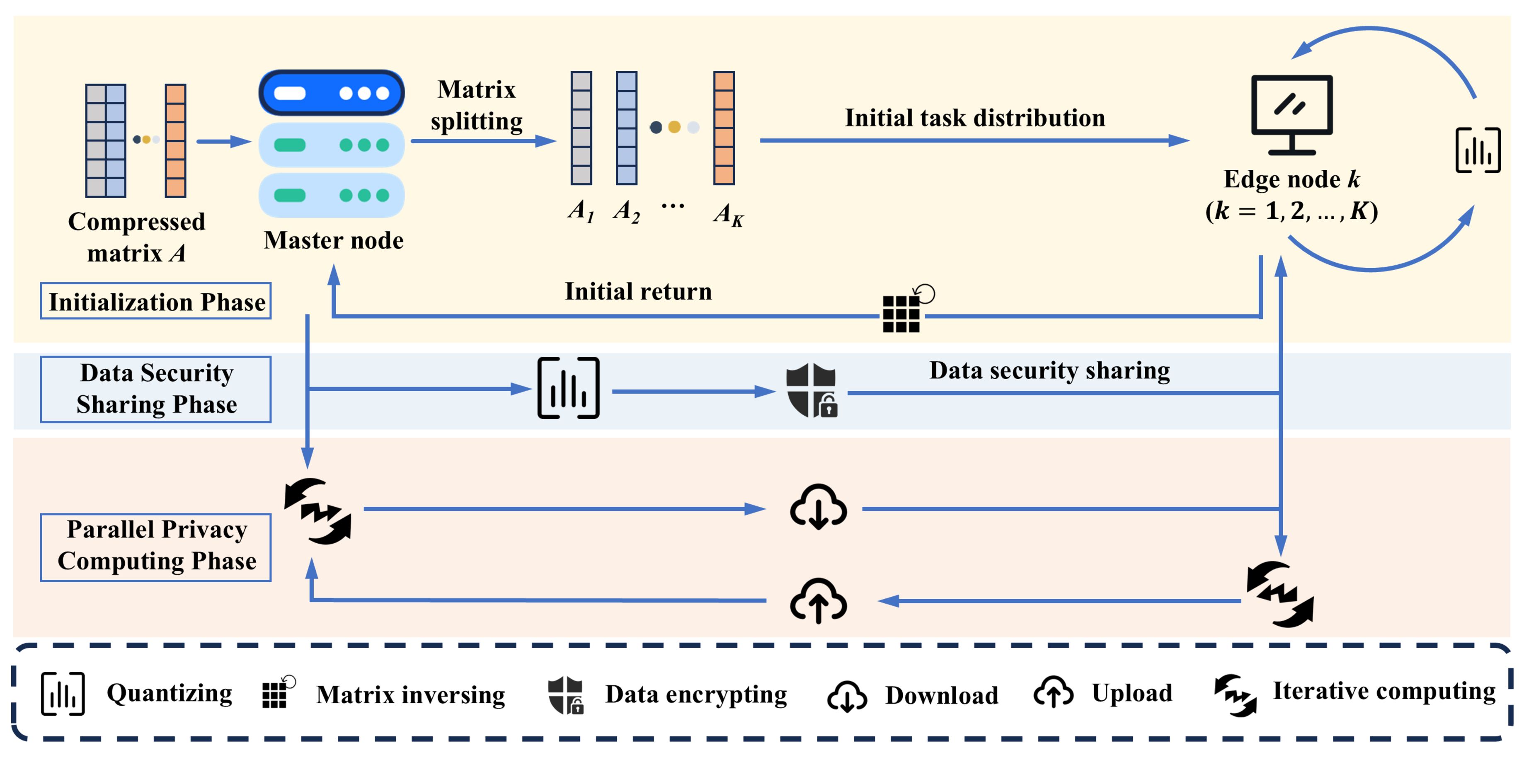}
\captionsetup{format=plain,justification=centering}
\caption{Proposed 3P-ADMM-PC2 scheme.}
\label{fig1}
\end{figure*}
\section{System Model and Problem Formulation}
Consider a distributed edge privacy-computing network, including one master node and $K$ incompletely trusted distributed nodes, where the master node borrows the computing power of other $K$ distributed nodes to compute a LASSO problem securely. In particular,
the LASSO problem is expressed as
\begin{equation}\label{sldo21}
\arg\min_{\mathbf{x}\in\mathbb{R}^N}\frac{1}{2}\left\| \mathbf{y}-\mathbf{A}\mathbf{x}  \right\|_{2}^{2}+\lambda {{\left\| \mathbf{x}  \right\|}_{1}},
\end{equation}
where $\mathbf{A}\in\mathbb{R}^{M\times N}(M\ll N)$ is the compressed matrix, $\mathbf{x}\in\mathbb{R}^N$ is a sparse vector. The vector $\mathbf{x}$ represents the global sparse signal to be jointly estimated from the observed data. In classical compressed sensing, $\mathbf{x}$ may denote the sparse coefficient vector of a natural image under a wavelet transform basis. In power grid reconstruction, $\mathbf{x}$ could represent the admittance or impedance parameters of a transmission line. $\mathbf{y}\in\mathbb{R}^M$ is the observation vector, $\lambda\geq 0$ is a regularization parameter.
Note that the problem \eqref{sldo21} can be further formulated as
\begin{subequations}\label{sldka}
\begin{align}
\min_{\mathbf{x} ,\mathbf{z}\in\mathbb{R}^N}\ & \frac{1}{2}\left\| \mathbf{y}-\mathbf{A}\mathbf{x}  \right\|_{2}^{2}+\lambda {{\left\| \mathbf{z}  \right\|}_{1}}, \\
&\mathrm{s.t.}\ \ \ \mathbf{x} -\mathbf{z} =0.
\end{align}
\end{subequations}
Then the corresponding augmented Lagrangian function is expressed as
\begin{align}\label{lakso2}
\mathcal{L}_{\rho }(\mathbf{x},\mathbf{z},\mathbf{u})&=\frac{1}{2}\left\| \mathbf{y}-\mathbf{A}\mathbf{x}  \right\|_{2}^{2}+\lambda {{\left\| \mathbf{z}  \right\|}_{1}}+{\mathbf{u}^{T}}(\mathbf{x} -\mathbf{z} ) \nonumber\\
  &+\frac{\rho }{2}\left\| \mathbf{x} -\mathbf{z}  \right\|_{2}^{2},
\end{align}
where $\mathbf{u}\in\mathbb{R}^N$ is a lagrangian multiplier vector, and $\rho\geq 0$ is a penalty parameter.
Based on the augmented Lagrangian function in \eqref{lakso2}, the iterative steps with respect to $\mathbf{x},\mathbf{z},\mathbf{v}$ for solving the LASSO problem can be expressed as
\begin{subequations}
\begin{align}
 &\mathbf{x}^{(t)}=(\mathbf{A}^T\mathbf{A}+\rho \mathbf{I}_N)^{-1}(\mathbf{A}^{T}\mathbf{y}+\rho (\mathbf{z}^{(t-1)}-\mathbf{v}^{(t-1)})),\label{sldo09}\\
 &\mathbf{z}^{(t)}={{S}_{\frac{\lambda }{\rho }}}(\mathbf{v}^{(t-1)}+\mathbf{x}^{(t)}),\label{ty32}\\
 &\mathbf{v}^{(t)}=\mathbf{v}^{(t-1)}+\mathbf{x}^{(t)}-\mathbf{z}^{(t)},\label{u78}
\end{align}
\end{subequations}
where $\mathbf{v}=\frac{1}{\rho}\mathbf{u}, \mathbf{I}_N \in \mathbb{R}^{N\times N}$ is a unit matrix, and $S_{\frac{\lambda}{\rho}}( \cdot )$ denotes the soft threshold function. It is important to point out that the complexity of \eqref{sldo09} is $\mathcal{O}(N^3)$. The complexity limits the application of ADMM, especially machine learning and large-scale signal processing.

To efficiently compute iterations of LASSO problem, the master node cooperates with $K$ distributed incompletely trusted nodes.
In particular, $\mathbf{x}$ is divided into $\mathbf{x}=(\mathbf{x}^T_1,\mathbf{x}^T_2,\cdots,\mathbf{x}^T_K)^T$, where $\mathbf{x}_k^T\in\mathbb{R}^{N_k}$, $N=\sum_{k=1}^KN_k$. The compressed matrix $\mathbf{A}$ is expressed as $\mathbf{A}=(\mathbf{A}_1,\mathbf{A}_2,\cdots,\mathbf{A}_K)$, where $\mathbf{A}_k\in\mathbb{R}^{M\times N_k}$. Then, the subproblem for solving $\mathbf{x}$ in \eqref{lakso2} is expressed as
\begin{align}\label{lop32}
&\!\!\!\!\arg\min_{\mathbf{x}}\frac{1}{2}\|\mathbf{y}-\mathbf{A}\mathbf{x}\|_2^2+\mathbf{u}^T(\mathbf{x}-\mathbf{z})+\frac{\rho}{2}\|\mathbf{x}-\mathbf{z}\|_2^2\nonumber\\
&\!\!\!\!=\arg\min_{\mathbf{x}}\frac{1}{2}\|\mathbf{y}-\mathbf{A}\mathbf{x}\|_2^2+\frac{\rho}{2}\left\|\mathbf{x}-\mathbf{z}+\frac{\mathbf{u}}{\rho}\right\|_2^2\nonumber\\
&\!\!\!\!=\arg\min_{\mathbf{x}_k}\frac{1}{2}\left\|\mathbf{y}-\!\!\sum_{k=1}^K\mathbf{A}_k\mathbf{x}_k\right\|_2^2\!\!\!\!+\frac{\rho}{2}\sum_{k=1}^K\left\|\mathbf{x}_k-\mathbf{z}_k+\!\frac{\mathbf{u}_k}{\rho}\right\|_2^2,
\end{align}
where we use the fact that $\mathbf{z}=(\mathbf{z}_1^T, \mathbf{z}_2^T, \cdots, \mathbf{z}_K^T)^T$ and $\mathbf{u}=(\mathbf{u}^T_1, \mathbf{u}^T_2, \cdots, \mathbf{u}_K^T)^T$.

Note that the first term of the objective function in \eqref{lop32} couplings $K$ variables, resulting in a low efficiency solution. Therefore, the upper bound of the first term in \eqref{lop32} is given by
\begin{align}
\left\|\mathbf{y}-\sum_{k=1}^K\mathbf{A}_k\mathbf{x}_k\right\|_2^2\leq \sum_{k=1}^K\left\|\frac{1}{K}\mathbf{y}-\mathbf{A}_k\mathbf{x}_k\right\|_2^2.
\end{align}

Problem \eqref{lop32} is further approximated as
\begin{align}\label{etlop32}
\arg\min_{\mathbf{x}_k}\frac{1}{2}\!\sum_{k=1}^K\left\|\frac{\mathbf{y}}{K}\!-\!\mathbf{A}_k\mathbf{x}_k\right\|_2^2\!+\frac{\rho}{2}\sum_{k=1}^K\left\|\mathbf{x}_k\!-\mathbf{z}_k+\!\frac{\mathbf{u}_k}{\rho}\right\|_2^2,
\end{align}
where the variables are separable. Thus, problem \eqref{etlop32} can be translated into $K$ subproblems:
\begin{align}\label{subetlop32}
\arg\min_{\mathbf{x}_k}\frac{1}{2}\!\left\|\frac{\mathbf{y}}{K}\!-\!\mathbf{A}_k\mathbf{x}_k\right\|_2^2\!+\frac{\rho}{2}\left\|\mathbf{x}_k\!-\mathbf{z}_k+\!\frac{\mathbf{u}_k}{\rho}\right\|_2^2, \forall k.
\end{align}

Similar to \eqref{sldo09}, the solution of \eqref{subetlop32} is expressed as
\begin{align}\label{sldo091}
\mathbf{x}_k^{(t)}\!=(\mathbf{A}_k^T\mathbf{A}_k+\rho \mathbf{I}_{N_k})^{-1}(\mathbf{A}_k^{T}\mathbf{y}+\rho (\mathbf{z}_k^{(t-1)}\!-\mathbf{v}_k^{(t-1)})).
\end{align}

According to \eqref{sldo091}, \eqref{ty32} and \eqref{u78}, the LASSO problem in \eqref{sldo21} can be solved by alternate iteration. To realize the synchronous computation, the iterations \eqref{sldo091}, \eqref{ty32} and \eqref{u78} are further expressed as follows:
\begin{subequations}\label{loi32}
\begin{align}
&\mathbf{x}_k^{(t)}=(\mathbf{A}_k^T\mathbf{A}_k+\rho \mathbf{I}_{N_k})^{-1}(\mathbf{A}_k^{T}\mathbf{y}+\rho (\mathbf{z}_k^{(t-1)}-\mathbf{v}_k^{(t-1)})), \forall k,\label{sldi231} \\
&\mathbf{z}^{(t)}={{S}_{\frac{\lambda }{\rho }}}(\mathbf{v}^{(t-1)}+\mathbf{x}^{(t-1)}), \label{sldo34}\\
&\mathbf{v}^{(t)}=\mathbf{v}^{(t-1)}+\mathbf{x}^{(t-1)}-\mathbf{z}^{(t)}.\label{iu43}
\end{align}
\end{subequations}

It is noted that the updates of $\mathbf{x}_k^{(t)}$, $\mathbf{z}^{(t)}$, and $\mathbf{v}^{(t)}$ in the $t$-th iteration is based on the $(t-1)$-th iteration.
As shown in Fig. \ref{fig1}, we propose a 3P-ADMM-PC2 scheme, which is summarized as follows:
\begin{itemize}
  \item {\bf{Initialization Phase}}: In this phase, the master node splits the computing task and sends them to different edge nodes. And the edge nodes return their initial calculated values to the master node.

  {\bf{Master node}}: The master node splits $\mathbf{A}, \mathbf{z}, \mathbf{v}$ in \eqref{sldo09} into $K$ parts, i.e., $\mathbf{A}=(\mathbf{A}_1,\mathbf{A}_2,\cdots, \mathbf{A}_K)$, $\mathbf{z}=(\mathbf{z}^T_1,\mathbf{z}^T_2,\cdots,\mathbf{z}^T_K)^T$, $\mathbf{v}=(\mathbf{v}^T_1,\mathbf{v}^T_2,\cdots,\mathbf{v}^T_K)^T$, and successively transmits  $\alpha_k=\{\mathbf{A}_k^{T}\mathbf{A}_k, \rho\}$ to edge node $k, k=1,2,\cdots, K$.

  {\bf{Edge node}}: Edge node $k$ computes $\mathbf{B}_k=(\mathbf{A}_k^T\mathbf{A}_k+\rho \mathbf{I}_{N_k})^{-1}\in\mathbb{R}^{N_k\times N_k}$ and returns it to the master node. Meanwhile,
  edge node $k$ saves the quantization of $\mathbf{B}_k\rho$ as $\bar{\mathbf{B}}_k$.
  \item {\bf{Data Security Sharing Phase}}: The observation $\mathbf{y}$ in \eqref{sldo21} is sensitive information. To prevent information leakage, the master node quantizes and encrypts the information containing $\mathbf{y}$ before sharing it.

       {\bf{Master node}}: The master node quantizes and encrypts $\mathbf{B}_k(\mathbf{A}_k^T\mathbf{y})\in\mathbb{R}^{N_k}$, i.e.,
      \begin{align}\label{sldot2}
      \hat{\boldsymbol{\alpha}}_k = f_{en}(\Gamma_{1}(\mathbf{B}_k\mathbf{A}_k^T\mathbf{y})\in\mathbb{Z}^{N_k}, \forall k,
      \end{align}
      where $\Gamma_{i}(\cdot)$ and $f_{en}(\cdot)$ are the quantization function and the encryption function, respectively. 

      {\bf{Edge node}}: The edge node $k$ downloads and saves $\hat{\boldsymbol{\alpha}}_k$ for the computing task.
  \item {\bf{Parallel Privacy-Computing Phase}}: In this phase, the information in \eqref{sldo09} is calculated and exchanged safely between the master node and the distributed edge nodes.

  {\bf{Master node}}: At the master node, $\mathbf{z}^{(t)}$ and $\mathbf{-v}^{(t)}$ are updated according to \eqref{sldo34} and \eqref{iu43} with $f_{de}(\hat{\mathbf{x}}^{(t)})$, where $f_{de}(\cdot)$ is the decryption function. Note that the updates of $\mathbf{z}^{(t)}, \mathbf{-v}^{(t-1)}$, and the collection of $\{\hat{\mathbf{x}}_k\}_{k=1}^K$. $\mathbf{-v}^{(t)}$ and $\hat{\mathbf{x}}^{(t)}$ require the $(t-1)$-th updates $\mathbf{z}^{(t-1)}$. To safely calculate $\mathbf{x}_k^{(t)}$ in \eqref{sldi231}, $\mathbf{z}^{(t)}$ and $\mathbf{-v}^{(t)}$ are quantized and encrypted, resulting in
      \begin{align}\label{uyt12}
      \hat{\mathbf{z}}_{k}^{(t)}=f_{en}(\Gamma_{2}(\mathbf{z}_k^{(t)})), \hat{\mathbf{v}}_k^{(t)}=f_{en}(\Gamma_{2}(-\mathbf{v}_k^{(t)})), \forall k,
      \end{align}
      and $\hat{\mathbf{z}}_{k}^{(t)}, \hat{\mathbf{v}}_k^{(t)}$ are transmitted to edge node $k$.

      {\bf{Edge node}}: At edge node $k$, a calculated value of $\mathbf{x}_k^{(t)}$ in \eqref{sldi231} is obtained based on $\hat{\mathbf{z}}_{k}^{(t-1)}$ and $\hat{\mathbf{v}}_k^{(t-1)}$, i.e.,
      \begin{align}\label{ldofi2}
      \hat{\mathbf{x}}_k^{(t)}=\hat{\boldsymbol{\alpha}}_k +\Gamma_{2}(\bar{\mathbf{B}}_k)(\hat{\mathbf{z}}_{k}^{(t-1)}+(-\hat{\mathbf{v}}_k^{(t-1)})).
      \end{align}
      Then $\hat{\mathbf{x}}_k^{(t)}$ is uploaded to the master node.
\end{itemize}

Note that $\mathbf{B}_k(\mathbf{A}_k^T\mathbf{y})$, $\mathbf{z}_{k}^{(t)}$, and $\mathbf{-v}_{k}^{(t)}$ are encrypted respectively at master node. Each edge node estimates $\hat{\mathbf{x}}_k^{(t)}$ in \eqref{ldofi2} with ciphertexts $\hat{\boldsymbol{\alpha}}_k$,
$\hat{\mathbf{z}}_k^{(t)}$, and $\hat{\mathbf{v}}_k^{(t)}$. It is important to point out that calculation of $\hat{\mathbf{x}}_k^{(t)}$ in \eqref{ldofi2} consists of the addition computation and the multiply-by-a-constant computation. In addition, $\mathbf{B}_k(\mathbf{A}_k^T\mathbf{y})$, $\mathbf{z}_{k}^{(t)}$, and $\mathbf{-v}_{k}^{(t)}$ are real vectors, whose elements are the real numbers. Especially, the positive-negative characteristic and floating-point types are challenge to encryption and efficient computing.

\section{Parallel Collaborative ADMM Privacy Computing}
In this section, we introduce proposed 3P-ADMM-PC2 scheme for distributed computation. In particular, a quantization is first designed by mapping the real number to a certain range of positive integers. Furthermore, we reveal that the designed quantizer does not affect the homomorphic properties of Paillier-based ADMM privacy computing.

\subsection{Quantization of Floating Point Type Data}
Note that the elements in $\mathbf{B}_k(\mathbf{A}_k^T\mathbf{y})$, $\bar{\mathbf{B}}_k$, $\mathbf{z}_{k}^{(t)}$, $\mathbf{-v}_{k}^{(t)}$ are the floating point type data, which can not be encrypted directly. Thus, we first quantize $\mathbf{B}_k(\mathbf{A}_k^T\mathbf{y})$, $\bar{\mathbf{B}}_k$, $\mathbf{z}_{k}^{(t)}$,
$\mathbf{-v}_{k}^{(t)}$, to $\Gamma_{1}(\mathbf{B}_k(\mathbf{A}_k^T\mathbf{y}))$, $\Gamma_{2}(\bar{\mathbf{B}}_k)$, $\Gamma_{2}(\mathbf{z}_{k}^{(t)})$, $\Gamma_{2}(\mathbf{-v}_{k}^{(t)})$, expressed as

\begin{subequations}
\begin{align}
\Gamma_{1}(\mathbf{B}_k(\mathbf{A}_k^T\mathbf{y})) &= \left\lfloor {\frac{\Delta^{2}(\mathbf{B}_k(\mathbf{A}_k^T\mathbf{y}) - (\mathbf{B}_k(\mathbf{A}_k^T\mathbf{y}))_{\min}\mathbf{1}_{N_k})}{((\mathbf{B}_k(\mathbf{A}_k^T\mathbf{y}))_{\max} - (\mathbf{B}_k(\mathbf{A}_k^T\mathbf{y}))_{\min})^{2}}} \right\rceil,\label{et21}\\
\Gamma_{2}(\bar{\mathbf{B}}_k) &= \left\lfloor {\frac{\Delta(\bar{\mathbf{B}}_k - (\bar{\mathbf{B}}_k)_{\min}\mathbf{1}_{N_k\times N_k})}{(\bar{\mathbf{B}}_k)_{\max} - (\bar{\mathbf{B}}_k)_{\min}}} \right\rceil,\\
\Gamma_{2}(\mathbf{z}_{k}^{(t)}) &= \left\lfloor {\frac{\Delta(\mathbf{z}_{k}^{(t)} - (\mathbf{z}_{k}^{(t)})_{\min}\mathbf{1}_{N_k})}{(\mathbf{z}_{k}^{(t)})_{\max} - (\mathbf{z}_{k}^{(t)})_{\min}}} \right\rceil,\\
\Gamma_{2}(\mathbf{-v}_{k}^{(t)}) &= \left\lfloor {\frac{\Delta(\mathbf{-v}_{k}^{(t)} - (-\mathbf{v}_{k}^{(t)})_{\min}\mathbf{1}_{N_k})}{(-\mathbf{v}_{k}^{(t)})_{\max} - (-\mathbf{v}_{k}^{(t)})_{\min}}} \right\rceil,
\end{align}
\end{subequations}
where $\mathbf{1}_{N_k}$ is all one vector, $\Delta>0$, $(\cdot)_{\max}$ denotes the maximum element of a vector or matrix, $(\cdot)_{\min}$ denotes the minimum element, and $\left\lfloor \cdot \right\rceil$ denotes rounding.

\begin{remark}
   The quantization function $\Gamma_{2}(\cdot)$ maps the element from the $N_k$ dimensional real number field $\mathbb{R}^{N_k}$ to an integer set $\{0,1,2,\cdots, \Delta\}^{N_k}$. Especially, the negative real numbers are mapped to non-negative integers, which can be encrypted by the Paillier algorithm. To ensure that the quantization scheme does not compromise the homomorphic properties of the subsequent encryption scheme, $\Gamma_{1}(\cdot)$, unlike $\Gamma_{2}(\cdot)$, maps the element from the $N_k$ dimensional real number field $\mathbb{R}^{N_k}$ to an integer set $\{0,1,2,\cdots, \lfloor \frac{\Delta^{2}}{max - min} \rceil \}^{N_k}$.
\end{remark}

\subsection{Encryption and Distributed ADMM Privacy Computation}

Note that the elements of $\mathbf{B}_k(\mathbf{A}_k^T\mathbf{y})$, $\mathbf{z}_{k}^{(t)}$, $\mathbf{-v}_{k}^{(t)}$ are quantized to the positive integers, then the  Paillier encryption and decryption algorithm can be used for the parallel collaborative ADMM privacy computing, where the key generation at the master node is presented as follows:
\begin{itemize}
    \item Choose two large prime numbers $p$, $q$.

    \item Computing $n = pq, \epsilon = \mathrm{lcm}(p - 1,q - 1)$.

    \item Select an integer $g \in \mathbb{Z}_{n^{2}}^{*}, \mathrm{gcd}(g, n)=1$.

    \item Select an integer $r \in \mathbb{Z}^{*}_{n}, \mathrm{gcd}(r, n)=1$.

    \item Define the function $\mathrm{L}(x) = \frac{x - 1}{n}$.

    \item Computing $\mu = (\mathrm{L}(g^{\epsilon}~\mathrm{mod}~n^{2}))^{- 1}~\mathrm{mod}~n$.
\end{itemize}
where $\mathrm{lcm}(a,b)$ denotes the least common multiple of $a$ and $b$, and $\mathrm{gcd}(a,b)$ denotes the greatest common divisor of $a$ and $b$.
The public key is $(n,g)$ and the private key is $(\epsilon, \mu)$.

At the master node, the quantized $\Gamma_{1}(\mathbf{B}_k(\mathbf{A}_k^T\mathbf{y}))$, $\Gamma_{2}(\mathbf{z}_{k}^{(t)})$, $\Gamma_{2}(\mathbf{-v}_{k}^{(t)})$ in \eqref{sldot2} and \eqref{uyt12} are encrypted as follows:
\begin{subequations}
\begin{align}
    \hat{\boldsymbol{\alpha}}_k&=f_{en}(\Gamma_{1}(\mathbf{B}_k(\mathbf{A}_k^T\mathbf{y})))=g^{\Gamma_{1}(\mathbf{B}_k(\mathbf{A}_k^T\mathbf{y}))}{\mathbf{r}}_{1}^{n} \:\mathrm{mod}\:n^{2},\\
    \hat{\mathbf{z}}_{k}^{(t)}&=f_{en}(\Gamma_{2}(\mathbf{z}_{k}^{(t)}))=g^{\Gamma_{2}(\mathbf{z}_{k}^{(t)})}{\mathbf{r}}_{2}^{n} \:\mathrm{mod}\:n^{2},\label{loi32}\\
   -\hat{\mathbf{v}}_{k}^{(t)}&=f_{en}(\Gamma_{2}(\mathbf{-v}_{k}^{(t)}))=g^{\Gamma_{2}(\mathbf{-v}_{k}^{(t)})}{\mathbf{r}}_{3}^{n} \:\mathrm{mod}\:n^{2},\label{loi321}
\end{align}
\label{itq1}
\end{subequations}

\noindent where ${\mathbf{r}}_{1}, {\mathbf{r}}_{2}, {\mathbf{r}}_{3} \in (\mathbb{Z}^{*}_{n})^{{N}_{k}}$. The elements of these sets correspond to different $r$ values, which are used for encrypting each plaintext.

It is important to point out that the Paillier algorithm is considered as a encryption function in \eqref{sldot2} and \eqref{uyt12} due to the homomorphic operations of ciphertexts are involved in \eqref{ldofi2}.
However, the quantization error will impact the estimation accuracy of $\hat{\mathbf{x}}_k^{(t)}$ in \eqref{ldofi2}, which is calculated based on ciphertext operations. In particular, calculation of $\hat{\mathbf{x}}_k^{(t)}$ in \eqref{ldofi2}
consists of the addition computation and the multiply-by-a-constant computation. Then we have the following homomorphism definitions of the parallel privacy computing phase of proposed 3P-ADMM-PC2 scheme.

\begin{definition}
   For any two ciphertexts $\mathbf{c}_{1} = f_{en}(\Gamma_{2}(\mathbf{z}_{k}^{(t)}))$, $\mathbf{c}_{2} = f_{en}(\Gamma_{2}(\mathbf{-v}_{k}^{(t)}))$, the addition operation $\oplus $ between ciphertexts is defined as
   \begin{align}\label{headd}
   f_{en}(\Gamma_{2}(\mathbf{z}_{k}^{(t)})) \oplus f_{en}(\Gamma_{2}(\mathbf{-v}_{k}^{(t)}))
   &= \mathbf{c}_{1}\mathbf{c}_{2}~\mathrm{mod}~n^{2} \nonumber\\
   = f_{en}(\Gamma_{2}(\mathbf{z}_{k}^{(t)}) & +\Gamma_{2}(\mathbf{-v}_{k}^{(t)})~\mathrm{mod}~n).
   \end{align}
\end{definition}

\begin{definition}
   For $\Gamma_{2}(\bar{\mathbf{B}}_k), \mathbf{c} = f_{en}(\Gamma_{2}(\mathbf{z}_{k}^{(t)})+\Gamma_{2}(\mathbf{-v}_{k}^{(t)}))$, the multiplication operation is defined as
   \begin{align}\label{hemul}
   \Gamma_{2}&(\bar{\mathbf{B}}_k) \otimes f_{en}(\Gamma_{2}(\mathbf{z}_{k}^{(t)})+\Gamma_{2}(\mathbf{-v}_{k}^{(t)}))
   = \mathbf{c}^{\Gamma_{2}(\bar{\mathbf{B}}_k)}\:\mathrm{mod}\:n^{2} \nonumber\\
   &= f_{en}(\Gamma_{2}(\bar{\mathbf{B}}_k)\cdot(\Gamma_{2}(\mathbf{z}_{k}^{(t)})+\Gamma_{2}(\mathbf{-v}_{k}^{(t)}))\:\mathrm{mod}\:n).
   \end{align}
\end{definition}

Based on Definition 1 and Definition 2, $\hat{\mathbf{x}}_k^{(t)}$ in \eqref{ldofi2} is calculated as
\begin{align}\label{itq2}
    &\hat{\mathbf{x}}_k^{(t)}=\hat{\boldsymbol{\alpha}}_k \oplus \Gamma_{2}(\bar{\mathbf{B}}_k) \otimes (\hat{\mathbf{z}}_{k}^{(t-1)}\oplus(-\hat{\mathbf{v}}_k^{(t-1)}))\nonumber\\
    &=\hat{\boldsymbol{\alpha}}_k \oplus (f_{en}(\Gamma_{2}(\bar{\mathbf{B}}_k)\cdot(\Gamma_{2}(\mathbf{z}_{k}^{(t)})+\Gamma_{2}(\mathbf{-v}_{k}^{(t)}))~\mathrm{mod} ~n))\nonumber\\
    &=f_{en}(\Gamma_{1}(\mathbf{B}_k(\mathbf{A}_k^T\mathbf{y}))+\Gamma_{2}(\bar{\mathbf{B}}_k)(\Gamma_{2}(\mathbf{z}_{k}^{(t)})\nonumber\\
    &+\Gamma_{2}(\mathbf{-v}_{k}^{(t)}))~\mathrm{mod}~n).
\end{align}

It is worth noting that the information for encryption in \eqref{itq2} is given by
\begin{align}\label{ret54}
\!\!\!\Gamma_{1}(\mathbf{B}_k(\mathbf{A}_k^T\mathbf{y}))+\!\Gamma_{2}(\bar{\mathbf{B}}_k)(\Gamma_{2}(\mathbf{z}_{k}^{(t)})+\!\Gamma_{2}(\mathbf{-v}_{k}^{(t)}))~\mathrm{mod}~n
\end{align}
containing two additions and one multiplication. Note that the elements of $\Gamma_{1}(\mathbf{B}_k(\mathbf{A}_k^T\mathbf{y}))$, $\Gamma_{2}(\bar{\mathbf{B}}_k)$, $\Gamma_{2}(\mathbf{z}_{k}^{(t)})$, $\Gamma_{2}(\mathbf{-v}_{k}^{(t)})$ are nonnegative integers.
The elements of $\mathbf{B}_k(\mathbf{A}_k^T\mathbf{y})$, $\bar{\mathbf{B}}_k$, $\mathbf{z}_{k}^{(t)}$, $\mathbf{-v}_{k}^{(t)}$ are real numbers. To evaluate the estimation performance of decryption and inverse quantization in \eqref{itq2}, we have the following theorem, which simplifies $\mathbf{B}_k(\mathbf{A}_k^T\mathbf{y})$, $\bar{\mathbf{B}}_k$, $\mathbf{z}_{k}^{(t)}$, $\mathbf{-v}_{k}^{(t)}$ to $\mathbf{u}_1, \mathbf{u}_2, \mathbf{u}_{3}\in\mathbb{R}^N, \mathbf{B}\in\mathbb{R}^{N\times N}$.
\begin{theorem}\label{lot}
For quantization $\Gamma_1, \Gamma_2$,  and $\mathbf{u}_1, \mathbf{u}_2, \mathbf{u}_{3}\in\mathbb{R}^N, \mathbf{B}\in\mathbb{R}^{N\times N}$, the quantization calculation in \eqref{ret54} can be given by
\begin{align}\label{yuyt2}
&\Gamma_{1}(\mathbf{u}_{3}) + \Gamma_2(\mathbf{B})(\Gamma_2(\mathbf{u}_{1})+\Gamma_2(\mathbf{u}_{2})) = \nonumber \\
&\left\lfloor \Delta^{2}
\frac{\mathbf{u}_{3}+\mathbf{B}(\mathbf{u}_{1} + \mathbf{u}_{2}) - (2\mathbf{B}\mathbf{1}_{N} + \mathbf{u}_{1} + \mathbf{u}_{2} + 1)z_{\min}}{(z_{\max} - z_{\min} )^{2}} \right. \nonumber \\
&  \left. + {{\Delta}^{2}\frac{2z_{\min}^{2}\mathbf{1}_{N}}{\left( z_{\max} - z_{\min} \right)^{2}}} \right\rceil,
\end{align}
which can be further approximate by
\begin{align}\label{loit3}
&\mathbf{u}_{3}+\mathbf{B}(\mathbf{u}_{1} + \mathbf{u}_{2}) \approx \nonumber \\
&\frac{(\Gamma_{1}(\mathbf{u}_{3}) + \Gamma_2(\mathbf{B})(\Gamma_2(\mathbf{u}_{1})+\Gamma_2(\mathbf{u}_{2})))(z_{\max} - z_{\min} )^{2}}{\Delta^{2}} \nonumber \nonumber \\
&+ (2\mathbf{B}\mathbf{1}_{N} + \mathbf{u}_{1} + \mathbf{u}_{2} + 1)z_{\min} - 2z_{\min}^{2}\mathbf{1}_{N}.
\end{align}
\end{theorem}
\begin{proof}
Without loss of generality, we use the close set $[z_{\min},z_{\max}]$ to approximate the open set $\mathbb{R}$. Then we have
\begin{align}
&[z_{\min},z_{\max}]^N\subset \mathbb{R}^N,\\
&[z_{\min},z_{\max}]^N\times [z_{\min},z_{\max}]^N\subset \mathbb{R}^{N\times N}.
\end{align}
The quantization of $\mathbf{u}_{1} \in \mathbb{R}^N$ can be expressed as
\begin{equation}
\Gamma_{2}(\mathbf{u}_{1}) = \left\lfloor {\Delta\frac{\mathbf{u}_{1} - z_{\min}\mathbf{1}_N}{z_{\max} - z_{\min}}} \right\rceil ,
\end{equation}
the quantization of $\mathbf{u}_{2} \in \mathbb{R}^N$ and $\mathbf{B} \in \mathbb{R}^{M \times N}$ is the same as that of $\mathbf{u}_{1}$. The quantization of $\mathbf{u}_{3} \in \mathbb{R}^N$ can be expressed as
\begin{equation}
\Gamma_{1}(\mathbf{u}_{3}) = \left\lfloor {{\Delta}^{2}\frac{\mathbf{u}_{3} - z_{\min}\mathbf{1}_N}{(z_{\max} - z_{\min})^{2}}} \right\rceil.
\end{equation}

The quantized $\Gamma_{2}(\mathbf{u}_{1})$, $\Gamma_{2}(\mathbf{u}_{2})$ perform an addition operation, resulting in
\begin{equation}\label{oit2}
\Gamma_{2}(\mathbf{u}_{1}) + \Gamma_{2}(\mathbf{u}_{2}) = \left\lfloor{\Delta\frac{\mathbf{u}_{1} + \mathbf{u}_{2} - 2z_{\min}\mathbf{1}_N}{z_{\max} - z_{\min}}} \right\rceil .
\end{equation}

Then, the quantized $\Gamma_{2}(\mathbf{B})$ performs the multiplication operation with \eqref{oit2}, we have
\begin{align}\label{toy6}
&\Gamma_2(\mathbf{B})(\Gamma_2(\mathbf{u}_{1})+\Gamma_2(\mathbf{u}_{2})) = \nonumber \\
&\left\lfloor \Delta^{2}
\frac{\mathbf{B}(\mathbf{u}_{1} + \mathbf{u}_{2}) - (2\mathbf{B}\mathbf{1}_{N} + \mathbf{u}_{1} + \mathbf{u}_{2})z_{\min} + 2z_{\min}^{2}\mathbf{1}_{N}}{(z_{\max} - z_{\min} )^{2}} \right\rceil.
\end{align}

Then, the quantized $\Gamma_{1}(\mathbf{u}_{3})$ performs the addition operation with \eqref{toy6}, resulting in \eqref{yuyt2}. Furthermore, the result in \eqref{loit3} can be obtained by \eqref{yuyt2}.
\end{proof}

The obtained results in Theorem \ref{lot} reveal that the quantized positive integers do not affect the homomorphic operation of Paillier encryption and the exact value of the original $\mathbf{u}_{3} + \mathbf{B}(\mathbf{u}_{1} + \mathbf{u}_{2})$ can be separated from the computational result in \eqref{yuyt2}.
Especially, the proposed quantization method can be used in conjunction with the Paillier homomorphic encryption scheme,
which can be applied in iterative computations under privacy protection.
According to Theorem \ref{lot}, quantizing \eqref{ldofi2} results in
\begin{align}
\label{quantizing_x}
\Gamma(\mathbf{x}_k^{(t)})
&= \Gamma_1\left(\mathbf{B}_k (\mathbf{A}_k^T \mathbf{y})\right) \nonumber \\
&\quad + \Gamma_2(\bar{\mathbf{B}}_k) \left( \Gamma_2(\mathbf{z}_k^{(t-1)}) + \Gamma_2(-\mathbf{v}_k^{(t-1)}) \right) \nonumber \\
&= \left\lfloor \Delta^2 \frac{\mathbf{B}_k (\mathbf{A}_k^T \mathbf{y}) - z_{\min} \mathbf{1}_N}{(z_{\max} - z_{\min})^2} \right. \nonumber \\
&\quad \left. + \Delta^2 \frac{(\bar{\mathbf{B}}_k - z_{\min} \mathbf{1}_N)(\mathbf{z}_k^{(t-1)} - \mathbf{v}_k^{(t-1)} - 2 z_{\min} \mathbf{1}_N)}{(z_{\max} - z_{\min})^2} \right\rceil.
\end{align}

\subsection{Decryption and Inverse Quantization}

In each iteration round, the edge node simultaneously sends $\hat{\mathbf{x}}_k^{(t)}$ to the master node, which can be decrypted as
\begin{equation}
   \Gamma(\mathbf{\mathbf{x}_k^{(t)}}) = \mathrm{L}(({{\hat{\mathbf{x}}}_{k}^{(t)}})^{\epsilon}\:\mathrm{mod}~{n}^{2})\mu~\mathrm{mod}~n.
\end{equation}

According to \eqref{loit3}, the inverse quantization of $\Gamma(\mathbf{\mathbf{x}_k^{(t)}})$ is given by
\begin{align}
  &\mathbf{x}_k^{(t)}\approx {\frac{\Gamma(\mathbf{x}_k^{(t)})}{{\Delta}^{2}}}(z_{\max} - z_{\min})^{2} \nonumber\\
  & +[2\bar{\mathbf{B}}_{k}\mathbf{1}_N+(\mathbf{z}_{k}^{(t-1)}\mathbf{-v}_{k}^{(t-1)})+1]z_{\min}-2z_{\min}^{2}\mathbf{1}_N.
\end{align}

\begin{remark}
   The quantization loss depends on $\Delta$ when other parameters are unchanged, which reduces with the increase of $\Delta$. Besides, selecting a larger key for encryption and decryption is aimed at ensuring security. For the Paillier encryption scheme, the condition for correct decryption of homomorphic computation results is the plaintext $m \in [0, n)$ corresponding to the ciphertext. In actual implementations, we select the quantization parameter $\Delta = 10^{15}$ to ensure that the impact of precision loss from inverse quantization on the algorithm is negligible. For $\Gamma_{1}(\cdot)$ and  $\Gamma_{2}(\cdot)$, the maximum quantized values are $\lfloor \frac{\Delta^{2}}{max - min} \rceil$ and $\Delta$, respectively. In this case, the maximum quantized values are converted from decimal to binary representation, with maximum lengths of 83 and 49 bits, respectively. In \eqref{quantizing_x}, the quantization result of $\Gamma_{2}(\cdot)$ has the greatest impact on the length of the computational result, as it undergoes multiple multiplication operations and successive additions. Therefore, in our analysis of the computational result length, we focus on the maximum length of $\Gamma_{2}(\cdot)$'s quantization result, which is 49 bits. The homomorphic computation based on \eqref{quantizing_x} yields a result with a maximum length of $99 + \log_{2^{(N_{k} + 1)}}$ bits. For a 2048-bit $n$, the iterative computation will not be overflowed from $[0, n)$.
\end{remark}

With the obtained ${\mathbf{x}_{k}^{(t)}}$, the master note updates $\mathbf{z}^{(t+1)}$ and $\mathbf{v}^{(t+1)}$ in \eqref{sldo34} and \eqref{iu43}, respectively, i.e.,
\begin{equation}
   {{\mathbf{z}}^{(t+1)}}={{S}_{\frac{\lambda }{\rho }}}\left\{ {{\mathbf{v}}^{(t)}}+{\Gamma}^{(-1)}(f_{de}(\mathbf{x})^{(t)}) \right\},
\end{equation}
\begin{equation}
{{\mathbf{v}}^{(t+1)}}={{\mathbf{v}}^{(t)}}+{\Gamma}^{(-1)}(f_{de}(\mathbf{x})^{(t)})-{{\mathbf{z}}^{(t)}},
\end{equation}
where $\mathbf{x}=(\mathbf{x}_1^T,\mathbf{x}^T,\cdots, \mathbf{x}_K^T)^T$, ${\Gamma}^{(-1)}(\cdot)$ is the inverse quantization function and $f_{de}(\cdot)$ is the decryption function.

Finally, the proposed 3P-ADMM-PC2 algorithm is summarized in Algorithm \ref{algorithmic1}. However, there exist numerous computations of encryption and decryption involving ModExp of large integers, which constitute the primary computational overhead.

\begin{algorithm}[tp]
    \caption{\textbf{Proposed 3P-ADMM-PC2}}
    \label{algorithmic1}
    \begin{algorithmic}[1]
    \Require $\mathbf{A}\in \mathbf{R}^{M \times N}, \mathbf{y} \in \mathbf{R}^{M}$
    \Ensure $\mathbf{x}^{(t)}$
        \State Initializing $\mathbf{x, z, v}$.
        \State Selection of appropriate $\rho, \lambda, \Delta$.
        \State Specifying the minimum and maximum values of the quantizing.
        \State Generating $p,q$.
        \State $n \gets pq$.
        \State Generating $g, r$.
        \State $\epsilon \gets lcm(p - 1, q - 1)$.
        \State $\mu \gets (\mathrm{L}(g^{\epsilon}~\mathrm{mod}~n^{2}))^{- 1}~\mathrm{mod}~n$.
        \State Master node : Spliting $\mathbf{A}$ by column; Sending $\alpha_{k},$ minimum and maximum values, $\rho, \Delta$ to edge node $k$.
        \State Edge node $k$ : Computing $\mathbf{B}_{k}$; Sends $\mathbf{B}_{k}$ to master node.
        \State Master node : Computing $\hat{\boldsymbol{\alpha}}_{k}$; Sends $\hat{\boldsymbol{\alpha}}_{k}$ to edge node $k$.
        \Statex Edge node $k$: Quantizes $\mathbf{B}_k\rho$ as $\bar{\mathbf{B}}_k$.
        \Statex \textbf{~~~~The first iteration is the same as the calculation in the loop}
        \For{$t = 1$ \textbf{to} $iter_{max}$}
               \State Edge node $k$: Sending $\hat{\mathbf{x}}_k^{(t-1)}$ to master node.
               \Statex ~~~~Master node: Sending $\hat{\mathbf{z}}_{k}^{(t-1)}$, $-\hat{\mathbf{v}}_{k}^{(t-1)}$ to edge node $k$.
               \State Edge node $k$: Computing $\hat{\mathbf{x}}_k^{(t)} \gets \hat{\boldsymbol{\alpha}}_k \oplus \Gamma(\bar{\mathbf{B}}_k) \otimes (\hat{\mathbf{z}}_{k}^{(t-1)}\oplus(-\hat{\mathbf{v}}_k^{(t-1)}))$.
               \Statex ~~~~Master node: Decrypting and inverse quantizing $\hat{\mathbf{x}}_k^{(t)}$, and splicing $\hat{\mathbf{x}}_k^{(t)}$ to get $\mathbf{x}^{(t - 1)}$, then computing $\mathbf{z}^{(t)} \gets S_{\frac{\lambda}{\rho}}(\mathbf{v}^{(t - 1)} + \mathbf{x}^{(t - 1)})$ and $\mathbf{v}^{(t)} \gets \mathbf{v}^{(t - 1)} + \mathbf{x}^{(t - 1)} - \mathbf{z}^{(t)}$.
        \EndFor
        \State \Return($\mathbf{x}^{(t)}$)
    \end{algorithmic}
\end{algorithm}

\section{Adaptive GPU Accelerated 3P-ADMM-PC2}

In this section, we propose an adaptive GPU accelerated 3P-ADMM-PC2 for large key space, which has at least 1024-bit key length. The encryption, the decryption and the homomorphic operations of proposed 3P-ADMM-PC2 refer to ModExp in large key space. In particular, the ModExp computation task in large key space is first decomposed into multiple ModExp computation subtasks in small key spaces. Then, the computational subtasks in smaller space are distributed to multiple streaming multiprocessors (SM) of GPU, which are processed in parallel. In addition, the ModExp of each subtask in SM is translated into integer multiplication and shift operation with adaptive bit width.

To decompose the ModExp of proposed 3P-ADMM-PC2 in large key space into multiple subtask in smaller spaces, the following lemmas are considered.

\begin{lemma} When the moduli in a system of congruences are pairwise coprime, the system has a unique solution modulo the product of all moduli, and its general solution is given by \cite{pei1996chinese}
  \begin{equation}
      {x}_{CRT} = {\sum\limits_{i = 1}^{num}{{a}_{i}{e}_{i}{T}_{i}~\mathrm{mod}~{T}}}.
  \end{equation}
\label{CRT}
\end{lemma}

\begin{lemma}
If ${p}^{2}, {q}^{2}$ are mutually prime, we have \cite{bezout1779theorie}
  \begin{equation}
      ({q}^{2})^{- 1}~({\mathrm{mod}~{p}^{2}}){q}^{2} + ({p}^{2})^{- 1}~({\mathrm{mod}~{q}^{2}}){p}^{2}~\mathrm{mod}~{n}^{2} = 1.
  \end{equation}
\label{Bezout}
\end{lemma}
By leveraging the Lemma \ref{CRT}, the computations in the encryption process (EP) performed in the large space ${n}^{2}$, are decomposed into two smaller spaces ${p}^{2}, {q}^{2}$. Consequently, the calculations for \eqref{itq1} are optimized as follows:

\noindent
$\mathbb{Z}_{{n}^{2}} \to \mathbb{Z}_{{p}^{2} \times {q}^{2}}:$
\begin{subequations}
\begin{align}
  {g}^{'} &= g~\mathrm{mod}~{p}^{2},\\
  {g}^{''} &= g~\mathrm{mod}~{q}^{2},\\
  \Gamma_{1}^{'}(\mathbf{B}_k(\mathbf{A}_k^T\mathbf{y}))&=\Gamma_{1}(\mathbf{B}_k(\mathbf{A}_k^T\mathbf{y}))~\mathrm{mod}~\varphi({p}^{2}),\\
  \Gamma_{1}^{''}(\mathbf{B}_k(\mathbf{A}_k^T\mathbf{y}))&=\Gamma_{1}(\mathbf{B}_k(\mathbf{A}_k^T\mathbf{y}))~\mathrm{mod}~\varphi({q}^{2}),\\
  \Gamma_{2}^{'}(\mathbf{z}_{k}^{(t)})&=\Gamma_{2}(\mathbf{z}_{k}^{(t)})~\mathrm{mod}~\varphi({p}^{2}),\\
  \Gamma_{2}^{''}(\mathbf{z}_{k}^{(t)})&=\Gamma_{2}(\mathbf{z}_{k}^{(t)})~\mathrm{mod}~\varphi({q}^{2}),\\
  \Gamma_{2}^{'}(\mathbf{-v}_{k}^{(t)})&=\Gamma_{2}(\mathbf{-v}_{k}^{(t)})~\mathrm{mod}~\varphi({p}^{2}),\\
  \Gamma_{2}^{''}(\mathbf{-v}_{k}^{(t)})&=\Gamma_{2}(\mathbf{-v}_{k}^{(t)})~\mathrm{mod}~\varphi({p}^{2}),
\end{align}
\end{subequations}
where $\varphi(\cdot)$ is the Euler's totient function. The corresponding calculation results are given by
\begin{subequations}
\begin{align}
  {({g}^{\Gamma_{1}(\mathbf{B}_k(\mathbf{A}_k^T\mathbf{y}))})}^{'}&={{g}^{'}}^{\Gamma_{1}^{'}(\mathbf{B}_k(\mathbf{A}_k^T\mathbf{y}))}~\mathrm{mod}~{p}^{2},\\
  {({g}^{\Gamma_{1}(\mathbf{B}_k(\mathbf{A}_k^T\mathbf{y}))})}^{''}&={{g}^{''}}^{\Gamma_{1}^{''}(\mathbf{B}_k(\mathbf{A}_k^T\mathbf{y}))}~\mathrm{mod}~{q}^{2},\\
   {({g}^{\Gamma_{2}(\mathbf{z}_{k}^{(t)})})}^{'}&={{g}^{'}}^{\Gamma_{2}^{'}(\mathbf{z}_{k}^{(t)})}~\mathrm{mod}~{p}^{2},\\
   {({g}^{\Gamma_{2}(\mathbf{z}_{k}^{(t)})})}^{''}&={{g}^{''}}^{\Gamma_{2}^{''}(\mathbf{z}_{k}^{(t)})}~\mathrm{mod}~{q}^{2},\\
   {({g}^{\Gamma_{2}(\mathbf{-v}_{k}^{(t)})})}^{'}&={{g}^{'}}^{\Gamma_{2}^{'}(\mathbf{-v}_{k}^{(t)})}~\mathrm{mod}~{p}^{2},\\
   {({g}^{\Gamma_{2}(\mathbf{-v}_{k}^{(t)})})}^{''}&={{g}^{''}}^{\Gamma_{2}^{''}(\mathbf{-v}_{k}^{(t)})}~\mathrm{mod}~{q}^{2}.
\end{align}
\end{subequations}
According to Lemma \ref{CRT}, we have
\begin{equation}\label{crt1}
\begin{aligned}
     g^{\Gamma_{1}(\mathbf{B}_k(\mathbf{A}_k^T\mathbf{y}))}~&\mathrm{mod}~{n}^{2}={({g}^{\Gamma_{1}(\mathbf{B}_k(\mathbf{A}_k^T\mathbf{y}))})}^{'}{({q}^{2})}^{-1}(~\mathrm{mod}~{p}^{2}){q}^{2}\\
     &+{({g}^{\Gamma_{1}(\mathbf{B}_k(\mathbf{A}_k^T\mathbf{y}))})}^{''}{({p}^{2})}^{-1}(~\mathrm{mod}~{q}^{2}){p}^{2}.
\end{aligned}
\end{equation}
Using Lemma \ref{Bezout}, \eqref{crt1} is further expressed as
\begin{equation}
\begin{aligned}
     g^{\Gamma_{1}(\mathbf{B}_k(\mathbf{A}_k^T\mathbf{y}))}~&\mathrm{mod}~{n}^{2}={({g}^{\Gamma_{1}(\mathbf{B}_k(\mathbf{A}_k^T\mathbf{y}))})}^{'}+[{({g}^{\Gamma_{1}(\mathbf{B}_k(\mathbf{A}_k^T\mathbf{y}))})}^{''}\\
     &-{({g}^{\Gamma_{1}(\mathbf{B}_k(\mathbf{A}_k^T\mathbf{y}))})}^{'}]{({p}^{2})}^{-1}(\mathrm{mod}~{q}^{2}){p}^{2}.
\end{aligned}
\label{gA-C-B}
\end{equation}
Similarly, \eqref{loi32} and \eqref{loi321} can be further calculated by
\begin{equation}
\begin{aligned}
     ({g}^{\Gamma_{2}(\mathbf{z}_{k}^{(t)})})~&\mathrm{mod}~{n}^{2}={({g}^{\Gamma_{2}(\mathbf{z}_{k}^{(t)})})}^{'}+[{({g}^{\Gamma_{2}(\mathbf{z}_{k}^{(t)})})}^{''}\\
     &-{({g}^{\Gamma_{2}(\mathbf{z}_{k}^{(t)})})}^{'}]{({p}^{2})}^{-1}(\mathrm{mod}~{q}^{2}){p}^{2},
\end{aligned}
\label{gz-C-B}
\end{equation}
\begin{equation}
\begin{aligned}
     ({g}^{\Gamma_{2}(\mathbf{-v}_{k}^{(t)})})~&\mathrm{mod}~{n}^{2}={({g}^{\Gamma_{2}(\mathbf{-v}_{k}^{(t)})})}^{'}+[{({g}^{\Gamma_{2}(\mathbf{-v}_{k}^{(t)})})}^{''}\\
     &-{({g}^{\Gamma_{2}(\mathbf{-v}_{k}^{(t)})})}^{'}]{({p}^{2})}^{-1}(\mathrm{mod}~{q}^{2}){p}^{2}.
\end{aligned}
\label{gv-C-B}
\end{equation}

Furthermore, the EP calculations for $\mu$ and ${\mathbf{c}}^{\epsilon}~\mathrm{mod}~{n}^{2}$ in the decryption phase can also be accelerated similar to that in encryption phase.

After optimizing the ModExp in the encryption and decryption processes, the computation overhead is reduced, and computations on the two smaller key spaces $\mathbb{Z}_{{p}^{2}}, \mathbb{Z}_{{q}^{2}}$, can be performed in parallel. However, the EP in large key space $\mathbb{Z}_{{n}^{2}}$ is decomposed into the smaller key spaces $\mathbb{Z}_{{p}^{2}}, \mathbb{Z}_{{q}^{2}}$, practical implementation remains challenging due to the limited computational power of edge nodes. For example, the key space $\mathbb{Z}_{{n}^{2}}^{4096bits}$ is merely decomposed into smaller spaces $\mathbb{Z}_{{p}^{2}}^{2048bits}, \mathbb{Z}_{{q}^{2}}^{2048bits}$ for a 2048-bit length key, which are still difficult to handle. In addition, the EP computations for a plaintext vector encrypting or decrypting always are processed one by one, increasing the computational overhead.
\vspace{-10pt}

\subsection{Parallel Computing in Streaming Multiprocessor}

 Due to the limited number of CPU cores, ModExp operations in large key spaces exhibit low parallelism on CPU. To implement in the hardware devices with high parallelism, such as GPU or FPGA, ModExp operations should be optimized due to  each GPU core supports only 32 bits or 64 bits operations. Note that the EP operation mainly includes multiplication of large integers and modular operations. Thus we present the GPU-accelerated deployment process of the proposed 3P-ADMM-PC2 scheme as follows.

\subsubsection{Multiplication of Large Integers}
The decimal large integers can be expressed by the corresponding polynomial coefficient vectors. Given a large decimal integer
\begin{equation}
      {D}_{(10)} = {\sum\limits_{w = 0}^{L -1}}{{d}_{w} \cdot {\tilde{b}}^{w}}.
\end{equation}
Using coefficient representation
\begin{equation}
      {\mathbf{D}}_{(\tilde{b})} = [{d}_{L -1}, \cdots, {d}_{0}].
\end{equation}
where $\tilde{b}$ denotes an arbitrary base. When $\tilde{b} = 2$, ${\mathbf{D}}_{(\tilde{b})}$ is the well-known binary bit string. Under coefficient representation, the multiplication of two large integers ${\mathbf{D}}_{1, (\tilde{b})}$ and ${\mathbf{D}}_{2, (\tilde{b})}$
\begin{equation}
\begin{aligned}\label{big_int_mul_1}
    [{\mathbf{D}}_{1} {\mathbf{D}}_{2}]_{(\tilde{b})} = [{\sum\limits_{w_{1} + w_{2} = 2L - 1}} {d_{{w}_{1}}d_{{w}_{2}}}, \cdots, {\sum\limits_{w_{1} + w_{2} = 0}} {d_{{w}_{1}}d_{{w}_{2}}}].
\end{aligned}
\end{equation}
Different $\tilde{b}$ can be chosen to allow multiple GPU cores to compute the product of each bit in parallel, and the results are stored in shared memory. Multiple computational tasks can even be assigned to a single GPU core, which stores the results in shared memory after computation to reduce access frequency and thus decrease computation time. Note that the multiplication of $[{\mathbf{D}}_{1}{\mathbf{D}}_{2}]_{(\tilde{b})}$ is equivalent to the multiplication of two polynomials. To further reduce the complexity, ${\mathbf{D}}_{ (\tilde{b})}$ is transformed into point-value representations using FFT, and performs dot product. The obtained results are subsequently converted back to coefficient representation using IFFT, i.e.,
\begin{equation}
    {\mathbf{D}}_{(\tilde{b}), FFT} = FFT({\mathbf{D}}_{(\tilde{b})}) = [{d}^{FFT}_{L - 1}, \cdots, {d}^{FFT}_{0}].
\label{FFT}
\end{equation}
Then the product of two numbers can be written as
\begin{equation}
\label{big_int_mul_2}
    [{\mathbf{D}}_{1}{\mathbf{D}}_{2}]_{(\tilde{b}), FFT} = [{d}^{FFT}_{1, L - 1}{d}^{FFT}_{2, L - 1}, \cdots, {d}^{FFT}_{1, 0}{d}^{FFT}_{2, 0}].
\end{equation}
Then, IFFT is applied to convert to coefficient representation
\begin{align}\label{IFFT}
    [{D}_{1}{D}_{2}]_{(10)} \Leftrightarrow IFFT([{\mathbf{D}}_{1}{\mathbf{D}}_{2}]_{(\tilde{b}), FFT}),
\end{align}

\begin{algorithm}[tp]
    \caption{\textbf{Proposed GPU-accelerated EP Computation}}
    \label{algorithmic2}
    \begin{algorithmic}[1]
    \Require $g, {\mathbf{m}}_{{(\tilde{b})}_{2}}, {\mathbf{n}}^{2}_{(\tilde{b})}, {n}^{2}.$
    \Ensure ${g}^{m}~\mathrm{mod}~{n}^{2}.$
        \State ${\mathbf{T}}_{(\tilde{b})} \gets 1$
        \State $length \gets len({\mathbf{n}}^{2}_{(\tilde{b})})$
        \State $R \gets \lfloor \frac{1 << 2length}{{n}^{2}} \rfloor$
        \State $\mathbf{R}_{(\tilde{b}), FFT} \gets FFT(\mathbf{R}_{(\tilde{b})})$
        \State ${\mathbf{n}}^{2}_{(\tilde{b}), FFT} \gets FFT({\mathbf{n}}^{2}_{(\tilde{b})})$

        \For{$i = 0$ \textbf{to} $len({\mathbf{m}}_{{(\tilde{b})}_{2}})$}
               \If{${\mathbf{m}}_{{(\tilde{b})}_{2}, i} == 1$}
                     \State ${\mathbf{T}}_{(\tilde{b}), FFT} \gets FFT({\mathbf{T}}_{(\tilde{b})})$
                     \State ${\mathbf{g}}_{(\tilde{b}), FFT} \gets FFT({\mathbf{g}}_{(\tilde{b})})$
                     \State ${\mathbf{a}}_{(\tilde{b}), FFT} \gets {\mathbf{T}}_{(\tilde{b}), FFT} \times {\mathbf{g}}_{(\tilde{b}), FFT}$
                     \State $\mathbf{Q}_{(\tilde{b})} \gets IFFT({\mathbf{a}}_{(\tilde{b}), FFT} \times \mathbf{R}_{(\tilde{b}), FFT}) >> 2length$
                     \State ${\mathbf{T}}_{(\tilde{b})} \gets IFFT({\mathbf{a}}_{(\tilde{b}), FFT} - FFT(\mathbf{Q}_{(\tilde{b})}) \times {\mathbf{n}}^{2}_{(\tilde{b}), FFT})$
                     \If{${\mathbf{T}}_{(\tilde{b})} >= {\mathbf{n}}^{2}_{(\tilde{b})}$}
                           \State ${\mathbf{T}}_{(\tilde{b})} \gets {\mathbf{T}}_{(\tilde{b})} - {\mathbf{n}}^{2}_{(\tilde{b})}$
                     \EndIf
               \EndIf
               \State ${\mathbf{g}}_{(\tilde{b}), FFT} \gets FFT({\mathbf{g}}_{(\tilde{b})})$
               \State ${\mathbf{b}}_{(\tilde{b}), FFT} \gets {\mathbf{g}}_{(\tilde{b}), FFT} \times {\mathbf{g}}_{(\tilde{b}), FFT}$
               \State $\mathbf{P}_{(\tilde{b})} \gets IFFT({\mathbf{b}}_{(\tilde{b}), FFT} \times \mathbf{R}_{(\tilde{b}), FFT}) >> 2length$
               \State ${\mathbf{g}}_{(\tilde{b})} \gets IFFT({\mathbf{b}}_{(\tilde{b}), FFT} - FFT(\mathbf{P}_{(\tilde{b})}) \times {\mathbf{n}}^{2}_{(\tilde{b}), FFT})$
               \If{${\mathbf{g}}_{(\tilde{b})} >= {\mathbf{n}}^{2}_{(\tilde{b})}$}
                     \State ${\mathbf{g}}_{(\tilde{b})} \gets {\mathbf{g}}_{(\tilde{b})} - {\mathbf{n}}^{2}_{(\tilde{b})}$
               \EndIf
        \EndFor
        \State \Return(${\mathbf{T}}_{(\tilde{b})} \to {T}_{10}$) \Comment Decimal conversion
    \end{algorithmic}
\end{algorithm}

In the computation of FFT and IFFT, $\xi$ serves as a primitive root on the cyclotomic polynomial, where $\xi = e^{-\frac{2 \pi i}{L}}$, $i$ denotes the imaginary unit, and $L$ must be a power of 2.

\begin{remark}
To ensure that $L$ is a power of 2, the $\tilde{b}$ can be chosen as 2, 4, 16, 256, 65536. Then, a 4096-bit large integer will be converted into vectors of lengths 4096, 2048, 1024, 512, 256, respectively. Furthermore, the dimensions of the Discrete Fourier Transform (DFT) matrix used in the subsequent FFT process are determined. And the inversion operation of DFT matrix for IFFT is performed only once.
\end{remark}
\vspace{-10pt}

\subsection{ModExp Computing Task Decomposition}
\begin{figure*}[htbp]
\centering
\includegraphics[width=7in]{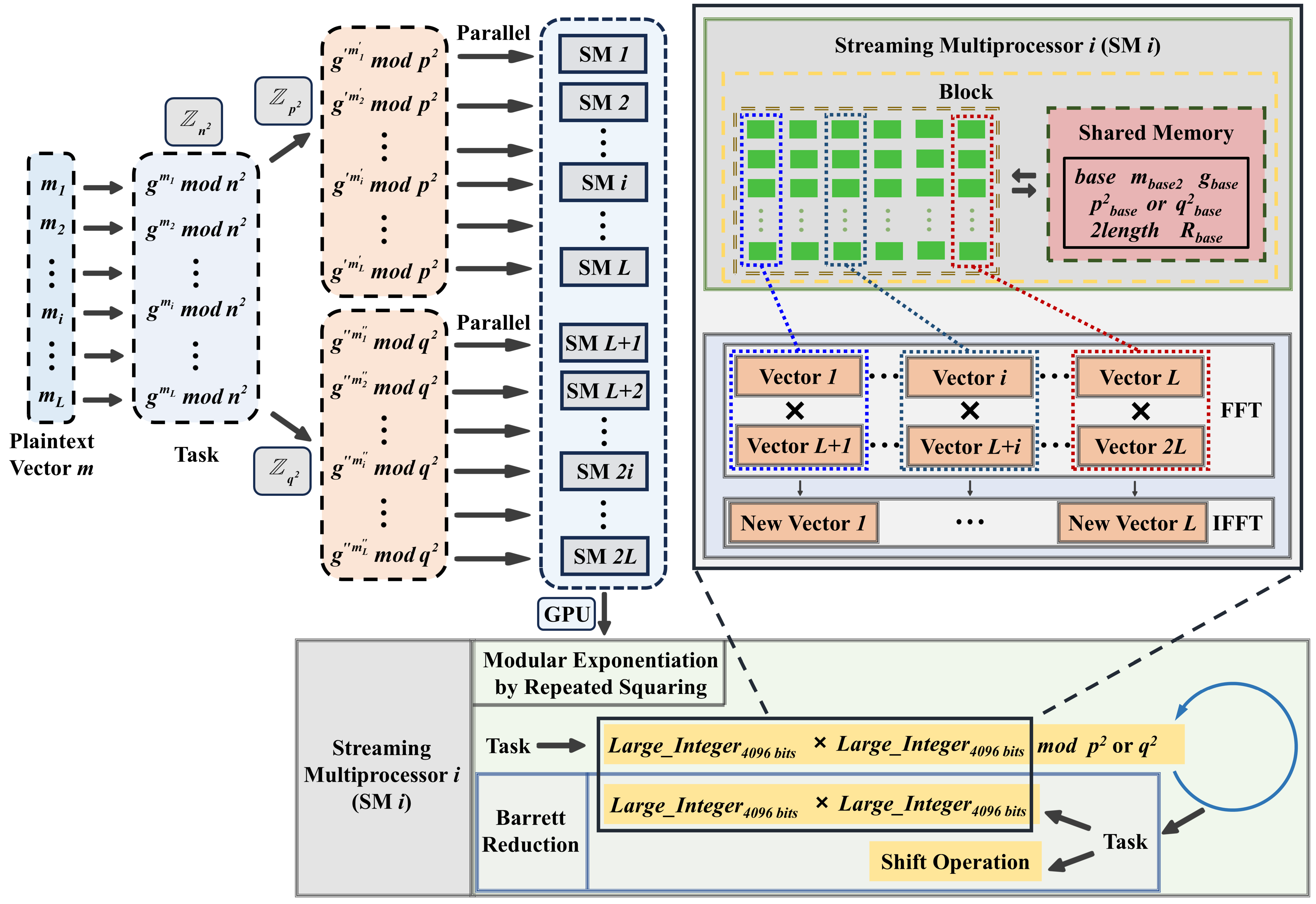}
\captionsetup{format=plain,justification=centering}
\caption{GPU-accelerated ModExp computation in large key space.}
\label{GPU-acce}
\end{figure*}

\subsubsection{Modular Operation}
To deploy ModExp in GPU, it is necessary to address the issue of overflow beyond the bit width of CUDA cores. Barrett Reduction replaces division with multiplication and shift operations. However, Barrett Reduction also involves large integer multiplication in its intermediate steps, which similarly overflowing the CUDA cores bit width. Based on aforementioned large integer multiplication, the large integer ModExp is optimized and decomposed. And our proposed GPU-accelerated EP computation is summarized in Algorithm \ref{algorithmic2}. The detailed computational architecture is illustrated in Fig. \ref{GPU-acce}.
\vspace{-10pt}

\subsection{Proposed GPU-accelerated 3P-ADMM-PC2}

Although GPU acceleration can be employed to expedite the encryption and decryption of plaintext vectors, the master node bears a substantial computational burden. Consequently, we optimize the compute burden of distributed nodes by sacrificing some communication overhead. A GPU-accelerated 3P-ADMM-PC2 algorithm is proposed, which involves leveraging edge nodes to collaboratively encrypt and decrypt data with the master node. During each iteration, each edge node communicates with the master node three times, exchanging different parameters, achieving a computation overhead balance for collaborative encryption and decryption, as shown in Fig. \ref{communication-three}.
\begin{figure}[tp]
\centering
\includegraphics[width=3.5in]{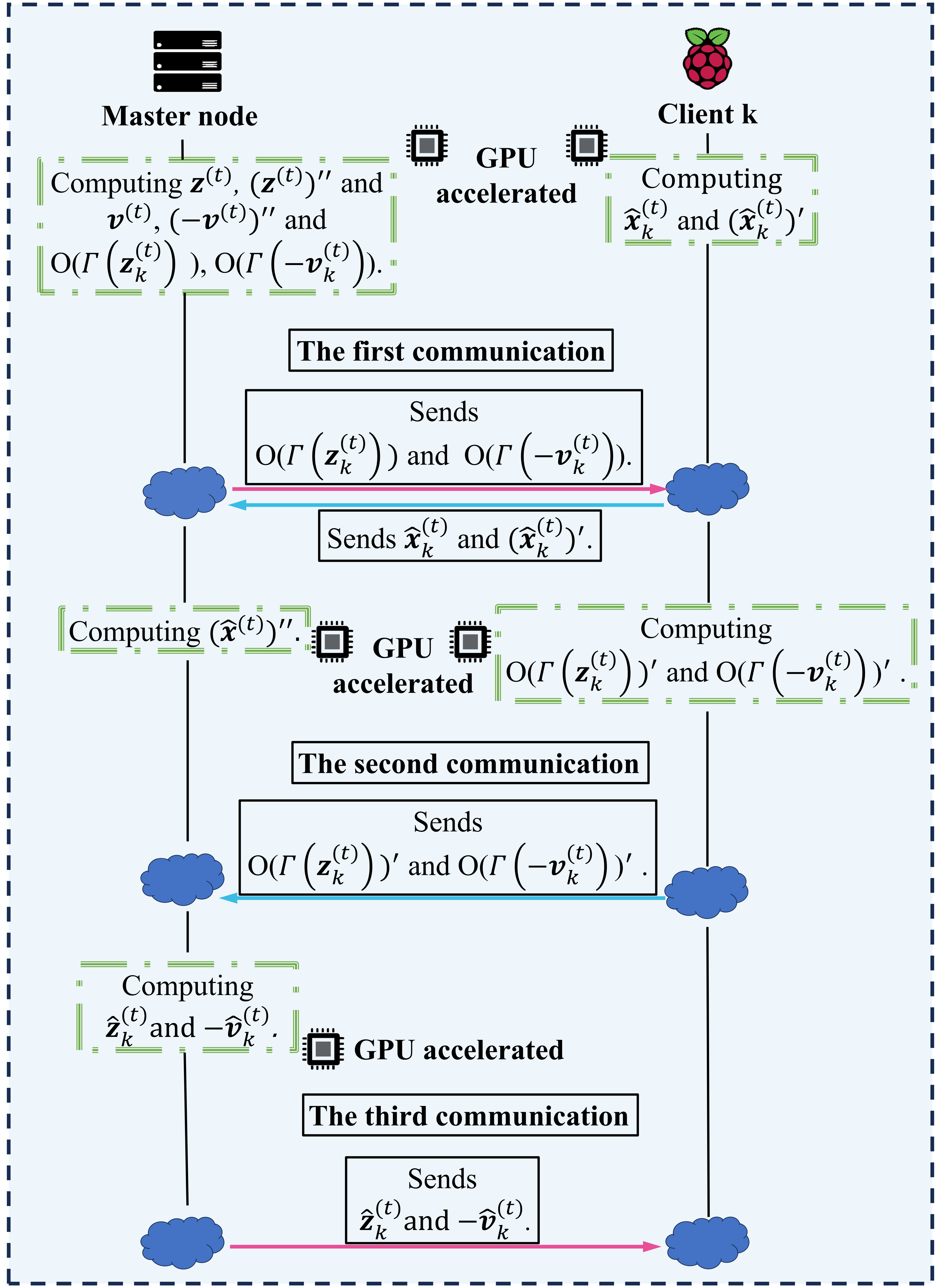}
\captionsetup{format=plain,justification=centering}
\caption{The illustration of proposed GPU-accelerated 3P-ADMM-PC2.}
\label{communication-three}
\end{figure}

Note that the calculation on large space $\mathbb{Z}_{{n}^{2}}$ can be decomposed into two smaller spaces $\mathbb{Z}_{{p}^{2}}, \mathbb{Z}_{{q}^{2}}$, and the results from these smaller spaces can be linearly combined back to the large space through Lemma \ref{CRT} and Lemma \ref{Bezout}, enabling parallel computations across two smaller spaces. Consequently, in our proposed GPU-accelerated 3P-ADMM-PC2 algorithm, the computations on $\mathbb{Z}_{{q}^{2}}$ are executed by the master node, while those on $\mathbb{Z}_{{p}^{2}}$ are handled by the edge nodes.
For one complete iteration of computation, the master node communicates with edge node $k$ three times. In the initialization parameter phase, the master node sends ${p}^{2}$ and $\varphi({p}^{2} )$ to edge node $k$ for subsequent computations.
After $\hat{\mathbf{x}}_{k}^{(t)}$ and $\mathbf{z}^{(t)}$, $\mathbf{v}^{(t)}$ are computed by the edge node $k$ and the master node,
the edge node $k$ and the master node then compute $(\hat{\mathbf{x}}_k^{(t)})^{'} = \hat{\mathbf{x}}_{k}^{(t)}~\mathrm{mod}~{{p}^{2}}$ and $(\mathbf{z}^{(t)})^{''} = \Gamma(\mathbf{z}^{(t)})~\mathrm{mod}~\varphi({q}^{2}), (\mathbf{-v}^{(t)})^{''} = \Gamma(\mathbf{-v}^{(t)})~\mathrm{mod}~\varphi({q}^{2})$. Three communication rounds of our proposed GPU-accelerated 3P-ADMM-PC2 algorithm can be presented as follows:

\subsubsection{The First Communication Round}

The master node sends $\mathbf{O}(\Gamma(\mathbf{z}_{k}^{(t)})), \mathbf{O}(\Gamma(\mathbf{-v}_{k}^{(t)}))$ to edge node $k$, where $\mathbf{O}(\cdot)$ is the obfuscation function. To ensure the privacy of the data,
master node allows the edge nodes $k$ to collaborate in the encryption, the data needs to be obfuscated and sent. At the same time, the edge node $k$ sends $\hat{\mathbf{x}}_{k}^{(t)}, (\hat{\mathbf{x}}_k^{(t)})^{'}$ to the master node.

After the first communication round, master node computes $(\hat{\mathbf{x}}_k^{(t)})^{''} = \hat{\mathbf{x}}_{k}^{(t)}~\mathrm{mod}~{{q}^{2}}$, and then computes $\mathbf{x}_{k}^{(t)}$ according to Lemma 1.
Edge node $k$ computes $(\mathbf{O}(\Gamma(\mathbf{z}_{k}^{(t)})))^{'} = \mathbf{O}(\Gamma(\mathbf{z}_{k}^{(t)}))~\mathrm{mod}~\varphi({p}^{2}), (\mathbf{O}(\Gamma(\mathbf{-v}_{k}^{(t)})))^{'} = \mathbf{O}(\Gamma(\mathbf{-v}_{k}^{(t)}))~\mathrm{mod}~\varphi({p}^{2})$.

\subsubsection{The Second Communication Round}

Edge node $k$ sends $(\mathbf{O}(\Gamma(\mathbf{z}_{k}^{(t)})))^{'}, (\mathbf{O}(\Gamma(\mathbf{-v}_{k}^{(t)})))^{'}$ to master node.

After the second communication round, master node calculates $\hat{\mathbf{z}}_{k}, -\hat{\mathbf{v}}_{k}$. according to (\ref{gz-C-B}) and (\ref{gv-C-B}).

\subsubsection{The Third Communication Round}

The master node sends $\hat{\mathbf{z}}_{k}, -\hat{\mathbf{v}}_{k}$ to edge node $k$.

Our proposed GPU-accelerated 3P-ADMM-PC2 algorithm is summarized in Algorithm \ref{algorithmic3}. The large integer ModExp of Algorithm \ref{algorithmic3} is performed according to Algorithm \ref{algorithmic2}.
\begin{algorithm}[tp]
    \caption{\textbf{Proposed GPU-accelerated 3P-ADMM-PC2}}
    \label{algorithmic3}
    \begin{algorithmic}[1]
    \Require $\mathbf{A}\in \mathbf{R}^{M \times N}, \mathbf{y} \in \mathbf{R}^{M}$
    \Ensure $\mathbf{x}^{(t)}$
        \State Initializing $\mathbf{x, z, v}$.
        \State Selection of appropriate $\rho, \lambda, \Delta$.
        \State Specifying the minimum and maximum values of the quantizing.
        \State Generating $p,q$.
        \State $n \gets pq$.
        \State Generating $g$.
        \State $\epsilon \gets lcm(p - 1, q - 1)$.
        \State $\mu \gets (\mathrm{L}(g^{\epsilon}~\mathrm{mod}~n^{2}))^{- 1}~\mathrm{mod}~n$.
        \State Master node: Spliting $\mathbf{A}$ by column; Sending $\alpha_{k},$ minimum and maximum values, $\rho, \Delta$ to edge node $k$.
        \State Edge node $k$: Computing $\mathbf{B}_{k}$; Sending $\mathbf{B}_{k}$ to master node.
        \State Master node: Computing $\hat{\boldsymbol{\alpha}}_{k}$; Sends $\hat{\boldsymbol{\alpha}}_{k}$ to edge node $k$.
        \Statex Edge node $k$: Quantizing $\mathbf{B}_k\rho$ as $\bar{\mathbf{B}}_k$.
        \For{$t = 1$ \textbf{to} $iter_{max}$}
               \State Edge node $k$: Sending $\hat{\mathbf{x}}_k^{(t-1)}, (\hat{\mathbf{x}}_k^{(t-1)})^{'}$ to master node.
               \Statex ~~~~Master node: Sending $\mathbf{O}(\Gamma(\mathbf{z}^{(t)})), \mathbf{O}(\Gamma(-\mathbf{v}^{(t)}))$ to edge node $k$.
               \State Edge node $k$: Computing $(\mathbf{O}(\Gamma(\mathbf{z}^{(t)})))^{'}$ and $(\mathbf{O}(\Gamma(-\mathbf{v}^{(t)})))^{'}$.
               \Statex ~~~~Master node: Computing $(\hat{\mathbf{x}}_k^{(t-1)})^{''}$ and $\mathbf{x}^{(t - 1)}$.
               \State Edge node $k$: Sending $(\mathbf{O}(\Gamma(\mathbf{z}^{(t)})))^{'}$ and $\mathbf{x}^{(t-1)}$ to master node.
               \State Master node: Computing $\hat{\mathbf{z}}_{k}, -\hat{\mathbf{v}}_{k}$.
               \State Master node: Sending $\hat{\mathbf{z}}_{k}, -\hat{\mathbf{v}}_{k}$ to edge node $k$.
               \State Edge node $k$: Computing $\hat{\mathbf{x}}_k^{(t-1)}, (\hat{\mathbf{x}}_k^{(t-1)})^{'}$.
               \Statex ~~~~Master node: Computing $\mathbf{z}^{(t)}, (\mathbf{z}^{(t)})^{''}$, $\mathbf{v}^{(t)}, (-\mathbf{v}^{(t)})^{''}$.
        \EndFor
        \State \Return($\mathbf{x}^{(t)}$)
    \end{algorithmic}
\end{algorithm}
\begin{remark}
   The master node has only sent ${p}^{2}$ to the edge nodes in our proposed GPU-accelerated 3P-ADMM-PC2 algorithm, and the edge nodes only know the ${p}^{2}$, $\varphi({p}^{2})$ without other parameters of the Paillier encryption and decryption. Thus the privacy of the data can still be guaranteed.
\end{remark}
\vspace{-10pt}

\subsection{Complexity and Convergence Analysis}
The complexity of Algorithm \ref{algorithmic3} is primarily determined by each large integer ModExp operation, making its essence the complexity of Algorithm \ref{algorithmic2}. In Algorithm \ref{algorithmic2}, \eqref{FFT}\eqref{big_int_mul_2}\eqref{IFFT} reduces the complexity of the multiplication calculation in \eqref{big_int_mul_1} from $\mathcal{O}({L}^{2})$ to $\mathcal{O}(LlogL)$. Ultimately, the overall complexity of Algorithm \ref{algorithmic2} is optimized to $\mathcal{O}(len({\mathbf{m}}_{{(\tilde{b})}_{2}})(5LlogL+2L))$, and further approximated to $\mathcal{O}(len({\mathbf{m}}_{{(\tilde{b})}_{2}})LlogL)$.

According to the convergence theory of approximate ADMM \cite{boyd2011distributed}\cite{eckstein1992douglas}, if the errors are bounded, the algorithm converges to a neighborhood of the optimal solution. In the distributed scenario presented in the paper, matrix $\mathbf{A} \in \mathbb{R}^{M \times N}$ is partitioned by columns into $[\mathbf{A}_{1}, \mathbf{A}_{2}, ..., \mathbf{A}_{K}]$ based on the total number $K$ of edge nodes. During computation, the data is quantized into positive integers, introducing quantization noise on the order of $10^{-14}$. Additionally, the matrix partitioning neglects off-diagonal blocks, which introduces a partitioning error. In practical algorithms, matrix $\mathbf{A}$ is a Gaussian matrix. Its column vectors $\mathbf{a}_{i}$ and $\mathbf{a}_{j}, i \neq j$ are random vectors, and their inner product $\mathbf{a}^{T}_{i}\mathbf{a}_{j}$ is a random variable satisfying the following conditions:
\begin{align}
E[\mathbf{a}_{i}^{T} \mathbf{a}_{j}] &= 0, \\
\text{Var}(\mathbf{a}_{i}^{T} \mathbf{a}_{j}) &= M.
\end{align}
where $E[\cdot]$ is expected value and $Var(\cdot)$ is the variance.Therefore, the variance of the normalized inner product $\frac{\mathbf{a}_{i}^{T} \mathbf{a}_{j}}{M}$ is $\frac{1}{M}$. When $M$ is large, the correlation between columns becomes negligible, indicating approximate orthogonality. For the segmentation error, if the columns of matrix $\mathbf{A}$ exhibit low correlation, the segmentation error remains bounded. Meanwhile, the quantization noise is on the order of $10^{-14}$ and is stably bounded. Therefore, the overall error remains small and bounded, ensuring the practical convergence of the algorithm.

\section{Experimental Results And Analysis}
\begin{figure}[htbp]
\centering
\includegraphics[width=3.5in]{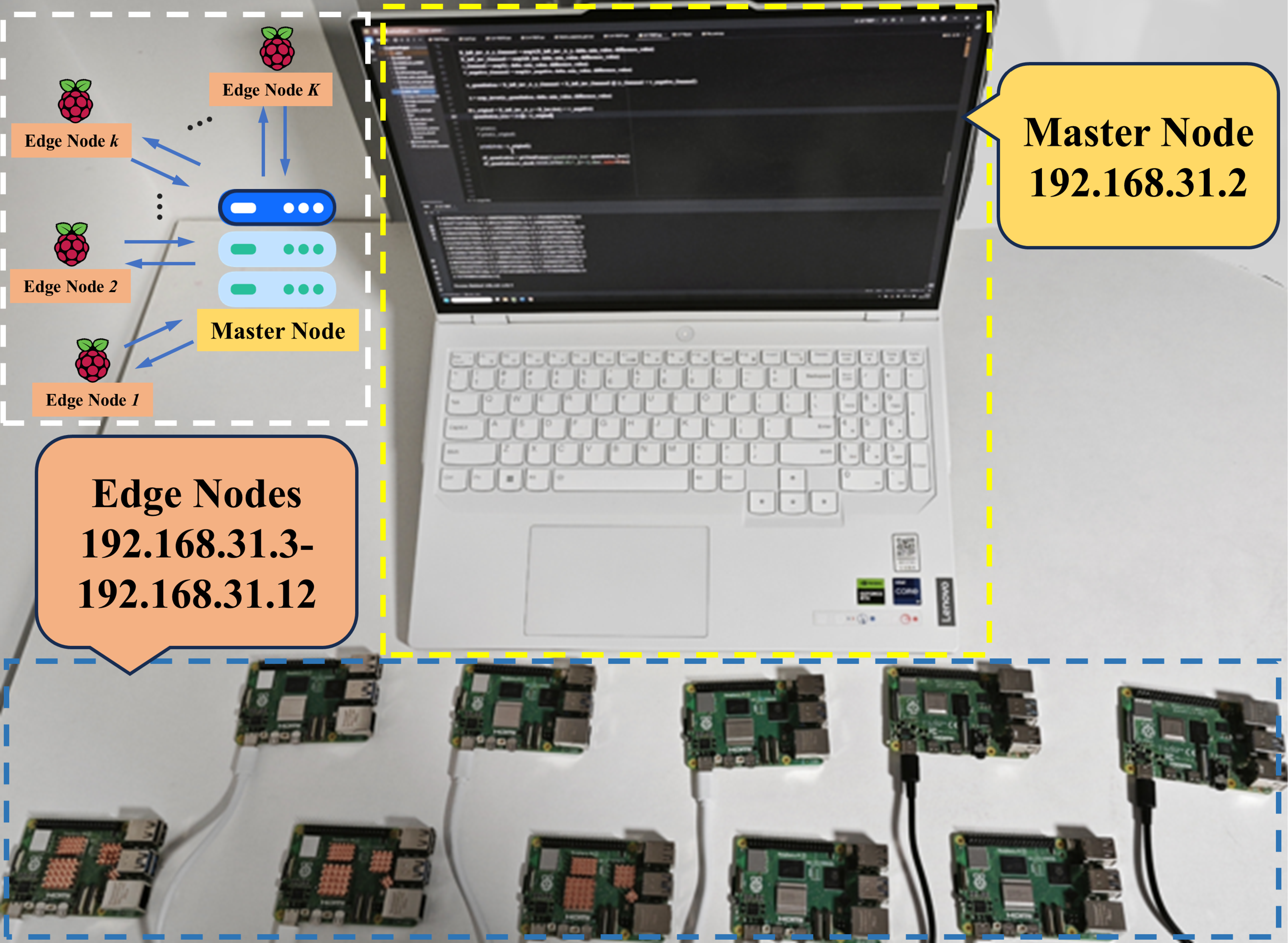}
\captionsetup{format=plain,justification=centering}
\caption{Hardware configuration and network environment for distributed computing.}
\label{hardware-environment}
\end{figure}
In this section, the proposed 3P-ADMM-PC2 scheme is first deployed in a hardware environment, and then the precision loss of the quantization as well as the mean squared error (MSE) are evaluated and compared. Furthermore, the proposed GPU-accelerated 3P-ADMM-PC2 scheme is detailed analyzed and evaluated. Finally, a typical application example for power network reconstruction is presented.

As shown in Fig. \ref{hardware-environment}, the hardware environment utilized in the experiments comprises a master node equipped with an Intel Core i9-13900HX CPU and 16GB RAM, along with an NVIDIA RTX 4060 8GB RAM GPU. The edge nodes consist of three Raspberry Pi 5 development boards, where each node is equipped with a Cortex-A76 64-bit SoC, 800MHz VideoCore VII GPU (supporting OpenGL ES 3.1 and Vulkan 1.2), and 8G LPDDR4X-4267 SDRAM.

In ADMM, $\rho$ is the penalty coefficient in the augmented Lagrangian function that controls constraint violation and penalty magnitude. $\lambda$ is the weighting coefficient of the L1 regularization term, which controls solution sparsity. Without loss of generality and to align with the data characteristics, we set both $\rho$ and $\lambda$ to 1 across all test scenarios.

\begin{figure*}[htbp]
\centering
\includegraphics[width=7in]{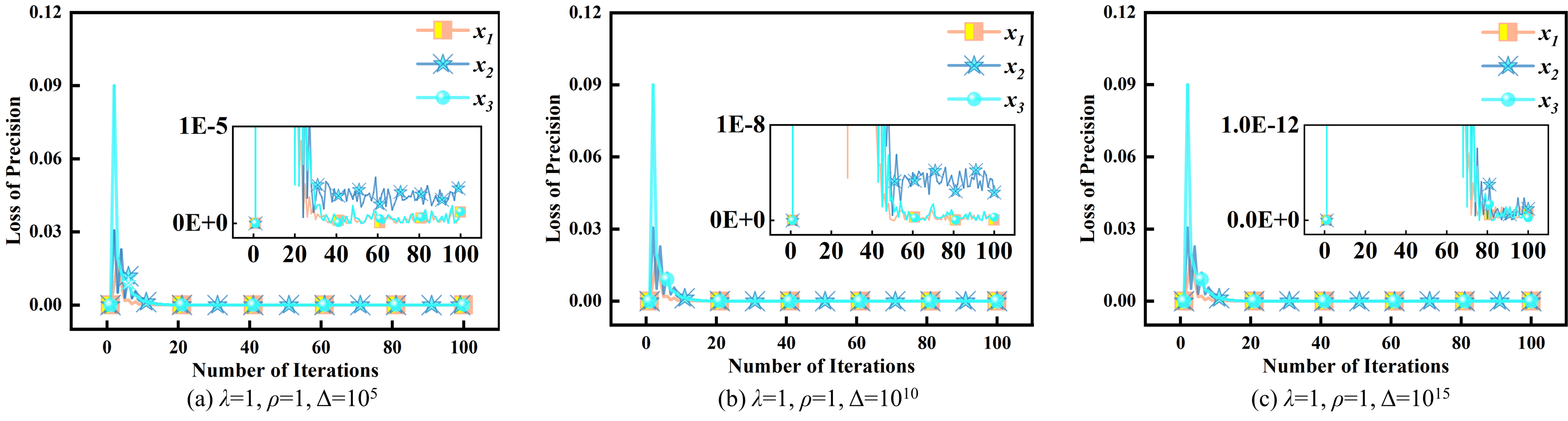}
\captionsetup{format=plain,justification=centering}
\caption{Precision loss of the quantization scheme with different $\Delta$ values.}
\label{x-loss-of-precision-image}
\end{figure*}
\begin{figure}[htbp]
\centering
\includegraphics[width=3.5in]{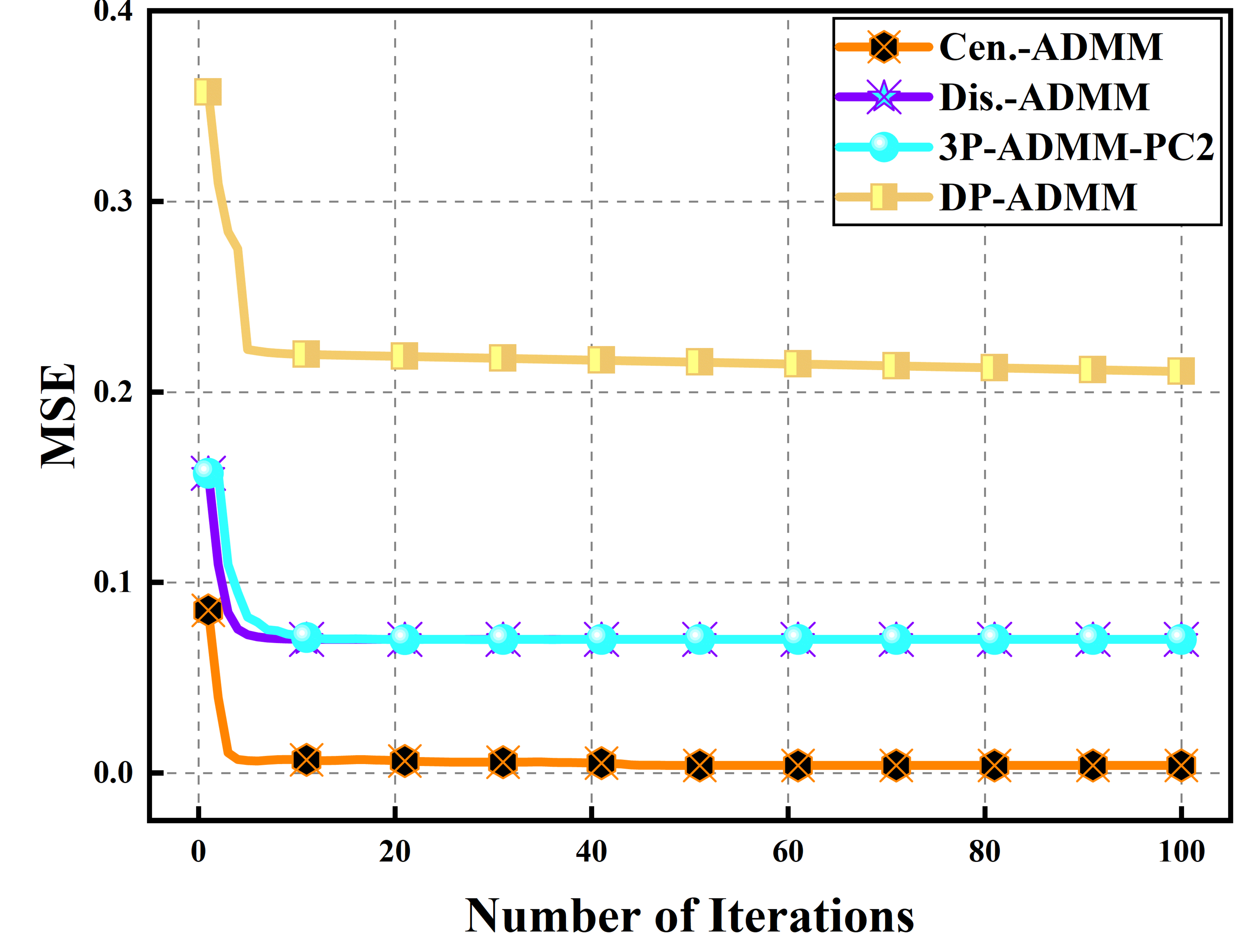}
\captionsetup{format=plain,justification=centering}
\caption{MSE of various methods.}
\label{MSE}
\end{figure}
\vspace{-10pt}

\subsection{Quantization Error and MSE of Proposed 3P-ADMM-PC2}
Fig. \ref{x-loss-of-precision-image} illustrates the quantization precision loss of the proposed 3P-ADMM-PC2. Without loss of generality, we consider $\mathbf{A} \in \mathbb{R}^{3 \times 3}, \mathbf{A} \sim \mathcal{CN}(0, 1)$. The precision loss is defined as the absolute value of the difference between the estimates of the unencrypted Dis.-ADMM and proposed 3P-ADMM-PC2 for $\mathbf{x}=(x_1,x_2,x_3)^T$. As shown in Fig. \ref{x-loss-of-precision-image}, as $\Delta$ increases from ${10}^{5}$ to ${10}^{15}$, the quantization precision loss decreases from ${10}^{-6}$ to ${10}^{-16}$. And the quantization precision loss can be approximated as $\frac{1}{10\Delta}$.

In a local area network (LAN), we use the hardware environment in Fig. \ref{hardware-environment} to deploy distributed computing algorithms. To facilitate a performance comparison with centralized ADMM, we selected the maximum dimension that centralized ADMM can handle, denoted as $\mathbf{A} \in \mathbb{R}^{3000 \times 27000}$, $\mathbf{A} \sim \mathcal{CN}(0,1)$. The computational tasks with $\mathbf{A}_{k} \in {\mathbb{R}}^{3000 \times 9000}, k = 1, 2, 3,$ are assigned to 3 edge nodes. We compare the MSE of the centralized ADMM (Cen.-ADMM), Dis.-ADMM, DP-ADMM, and proposed 3P-ADMM-PC2. In the proposed 3P-ADMM-PC2, the key length is 2048 bits with $\Delta = {10}^{15}$. The choice of $\mathbf{A}$'s dimension is based on the dimension of $\mathbf{A}^T\mathbf{A}+\rho \mathbf{I}_N$. When inverting $\mathbf{A}^T\mathbf{A}+\rho \mathbf{I}_N$ with dimension $27000 \times 27000$, approximately 14 GB of memory is required, which is the limit of available memory for Cen.-ADMM. In Fig. \ref{MSE}, the Cen.-ADMM shows the best MSE performance. The Dis.-ADMM has an MSE approximately 0.07 higher than that of Cen.-ADMM. The DP-Dis.-ADMM has an MSE approximately 0.2 higher than that of Cen.-ADMM. The MSE of the proposed 3P-ADMM-PC2 differs from that of the Dis-ADMM by approximately ${10}^{-14}$. The MSEs of both schemes almost overlap, indicating that the quantization precision loss is negligible.

To evaluate the performance of the proposed 3P-ADMM-PC2 under different numbers of edge nodes, we compared the MSE for configurations with 3 and 10 edge nodes. Simultaneously, to accommodate large-scale sparse signal processing, we increased the dimension of $\mathbf{x}^{(t)}$ to $\mathbf{x}^{(t)} \in \mathbb{R}^{65536}$. Accordingly, the dimension of matrix $\mathbf{A}$ was expanded to $\mathbf{A} \in \mathbb{R}^{10000 \times 65536}$. The computational tasks $\mathbf{A}_{k} \in \mathbb{R}^{10000 \times 21846}, k = 1, 2, 3,$ and $\mathbf{A}_{k} \in \mathbb{R}^{10000 \times 6554}, k = 1, 2, \cdots, 10,$ are distributed to 3 and 10 edge nodes, respectively. As shown in Fig. \ref{3-10-nodes-MSE}, we compared the convergence rate and MSE of the algorithm under conditions of 3 and 10 edge nodes, respectively, at sparsity levels of 10\%, 30\%, 50\%, 70\%, and 90\% for $\mathbf{x}^{(t)}$. It can be observed that when the sparsity is high, the algorithm converges faster with a lower MSE. As sparsity decreases, the convergence speed decreases significantly, and the MSE rises accordingly. This aligns with the characteristics of ADMM, which introduces an L1 regularization term to promote sparsity directly through a soft-thresholding shrinkage operator. When the original data is relatively sparse, the algorithm converges more quickly and approximates the true solution more easily, resulting in smaller errors. The number of edge nodes also impacts the algorithm's performance: increasing the number of edge nodes accelerates the overall computation but slightly degrades performance. Under the same sparsity level, the configuration with 3 edge nodes achieves lower errors. In practical applications, it is necessary to balance the requirements of computational speed and accuracy.
\begin{figure}[htbp]
\centering
\includegraphics[width=3.3in]{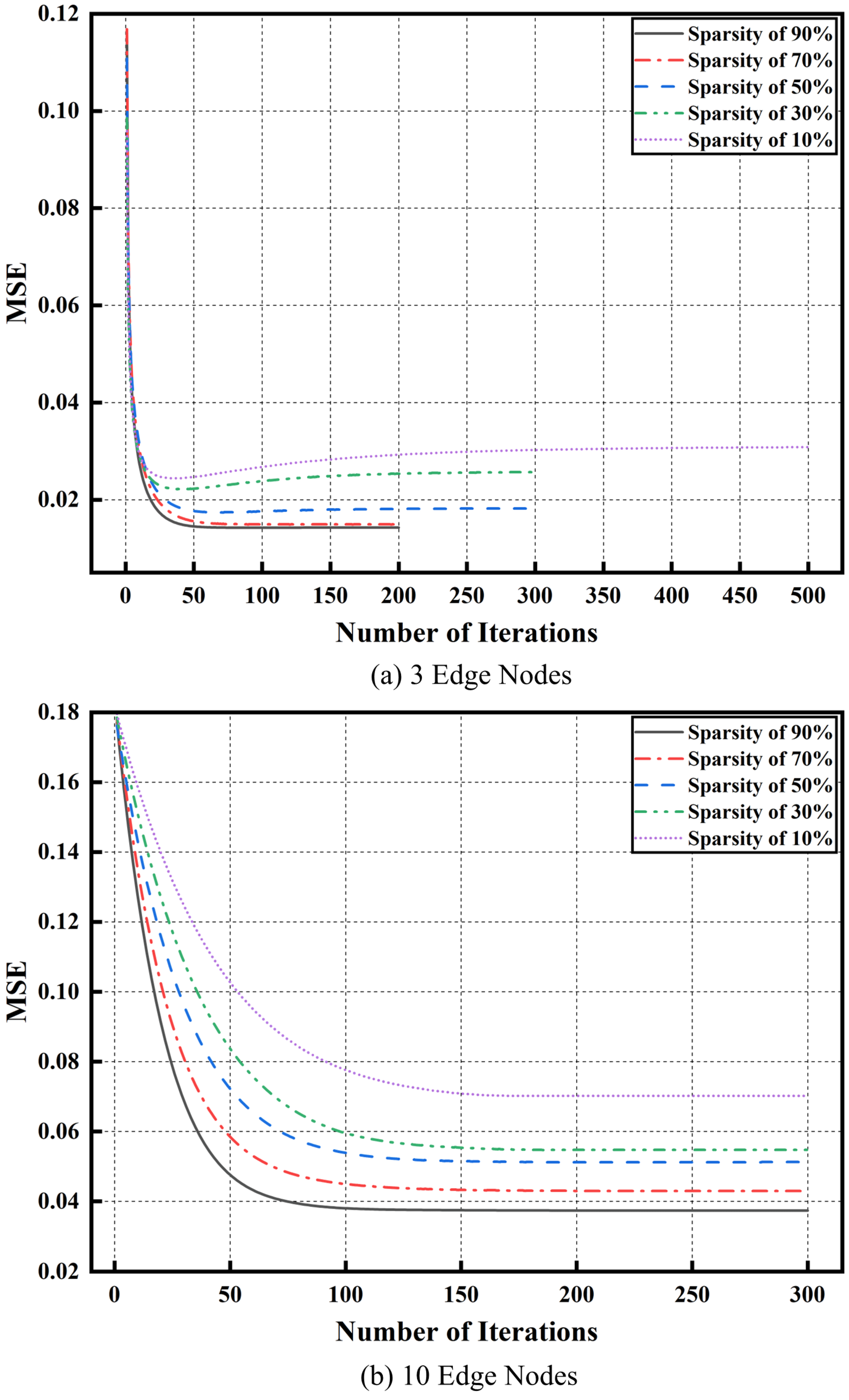}
\captionsetup{format=plain,justification=centering}
\caption{MSE with 3 and 10 edge nodes under varying sparsity levels.}
\label{3-10-nodes-MSE}
\end{figure}

\begin{table*}[htbp]
\centering
\caption{Throughput of CPU and GPU computations on the master node and edge node with different key lengths.}
\label{OPS}
\begin{tabularx}{7in}{|>{\centering\arraybackslash}p{1in}|>{\centering\arraybackslash}X|>{\centering\arraybackslash}X|>{\centering\arraybackslash}X|>{\centering\arraybackslash}X|>{\centering\arraybackslash}X|>{\centering\arraybackslash}X|>{\centering\arraybackslash}X|
                       >{\centering\arraybackslash}X|>{\centering\arraybackslash}X|}
\toprule
 \centering Throughput & \multicolumn{3}{c|}{ModMult(KOPS)} & \multicolumn{3}{c|}{ModExp(OPS)} & \multicolumn{3}{c|} {GPU-accelerated EP computation(OPS)} \\
 \cmidrule(lr){1-1} \cmidrule(lr){2-4} \cmidrule(lr){5-7} \cmidrule(lr){8-10}
 \centering Length of key n (bits) & 1024 & 2048 & 4096 & 1024 & 2048 & 4096 & 1024 & 2048 & 4096 \\
\midrule
Matser Node (CPU) & 207.95 & 42.37 & \textbf{16.97} & 810.2 & 113.38 & \textbf{16.56} & 1352.4 & 211.59 & \textbf{31.29} \\
Master Node(GPU) & 4062.66 & 1544.13 & \textbf{1414.34} & 79352.5 & 31080.16 & \textbf{26594.01} & 137807.85 & 65662.72 & \textbf{50053.47} \\
Edge Node (CPU) & \textbf{41.27} & 11.17 & 3.09 & 53.59 & \textbf{7.26} & 1.01 & 119.69 & \textbf{22.59} & 4.49 \\
Edge Node (GPU) & \textbf{272.39} & 67.23 & 5.05 & 9729.45 & \textbf{2553.25} & 260.98 & 22626.51 & \textbf{7957.05} & 701.75 \\
\bottomrule
\end{tabularx}
\end{table*}

\subsection{Performance Evaluation of Proposed GPU-accelerated 3P-ADMM-PC2}

\subsubsection{Throughput}
The core computation of proposed GPU-accelerated 3P-ADMM-PC2 is proposed GPU-accelerated EP computation in Algorithm \ref{algorithmic2}, which involves modular multiplication (ModMult) and ModExp calculations for large numbers. Operations Per Second (OPS) is used to demonstrate the throughout of main computations on the master node and edge nodes for proposed GPU-accelerated 3P-ADMM-PC2. TABLE \ref{OPS} shows the OPS for the computations with different key lengths, calculated on the CPU and GPU of the master node and edge nodes. Nine thousand unsigned integer plaintexts within the range $[0, n)$ were selected for each trial, repeated ten times, and the averages computed, with the modulus for the modulo operation set to $n^{2}$. It is observed that as the length of the public key $n$ increases, the computational load increases, leading to more frequent communication between GPU's threads and a corresponding decrease in throughput.

\subsubsection{Total Computation Time}
From the beginning of iteration to stable convergence, the process can be divided into two phases: pre-computation and iterative calculation. Each iteration includes local computation and communication. Therefore, the total computation time is given by
\begin{align}
    T_{total} = T_{pre} + (T_{loc} + T_{comm}) \times Iterations,
\end{align}
where $T_{pre}$ denotes the pre-computation time, $T_{loc}$ and $T_{comm}$represents the local computation time and communication times for each iteration round. When $\mathbf{A} \in \mathbb{R}^{3000 \times 27000}$ and 3 edge nodes, Fig. \ref{com_times} displays $T_{pre}$ and $T_{total}$ for different key lengths using the Cen.-ADMM, Dis.-ADMM, the CPU-Dis.-ADMM with CPU based encryption and decryption, and the proposed GPU-accelerated 3P-ADMM-PC2. Both Cen.-ADMM and Dis.-ADMM do not involve encryption and decryption computations, thus the changes of key lengths do not affect the computation time. The Dis.-ADMM has the shortest $T_{pre}$ and $T_{total}$ among all schemes but lacks privacy protection measures during data exchange. For 2048-bit key length, the difference in $T_{total}$ between the CPU-Dis.-ADMM and the Cen.-ADMM is small. The $T_{total}$ for the CPU-Dis.-ADMM is significantly larger than that for the Cen.-ADMM for 4096-bit key length. Thus, the distributed computing scheme using CPU for encryption and decryption becomes impractical. Moreover, the proposed GPU-accelerated 3P-ADMM-PC2 demonstrates better acceleration effects for $T_{total}$ compared to the Cen.-ADMM across different key lengths. Additionally, the proposed GPU-accelerated 3P-ADMM-PC2 provides greater security than the Dis.-ADMM during data exchange.

\subsubsection{Calculation Delay}
When $\mathbf{A} \in \mathbb{R}^{3000 \times 27000}$ and 3 edge nodes, Fig. \ref{com-delay} illustrates the waiting time between nodes during iterative calculations in the proposed GPU-accelerated 3P-ADMM-PC2 with different key lengths. Without loss of generality, we analyze the calculation delay for the 30th and the 80th iterations. As shown in Fig. \ref{com-delay} (a), the master node completes local computations at 9427 seconds and 16721 seconds for the 30th and the 80th iterations, respectively. The edge nodes finish their computations approximately at 9439 seconds and 16733 seconds, respectively. The master node waits approximately 12 seconds for the edge nodes to complete each iteration. Edge nodes mutually wait about 1.5 seconds for each other. The waiting time for the master node is primarily due to the difference in hardware conditions between the master node and edge nodes. The waiting time among edge nodes is caused by uneven plaintext lengths. Especially, as the key length increases, the mutual waiting time between nodes also increases.

\begin{figure*}[htbp]
\centering
\includegraphics[width=7in]{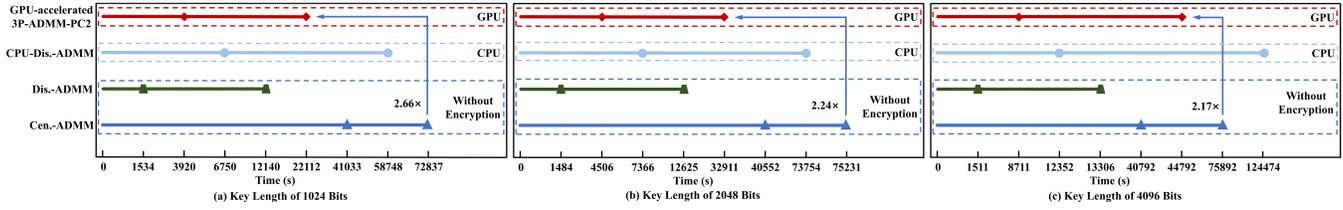}
\captionsetup{format=plain,justification=centering}
\caption{Computational time of various schemes with different keys.}
\label{com_times}
\end{figure*}

\begin{figure*}[htbp]
\centering
\includegraphics[width=7in]{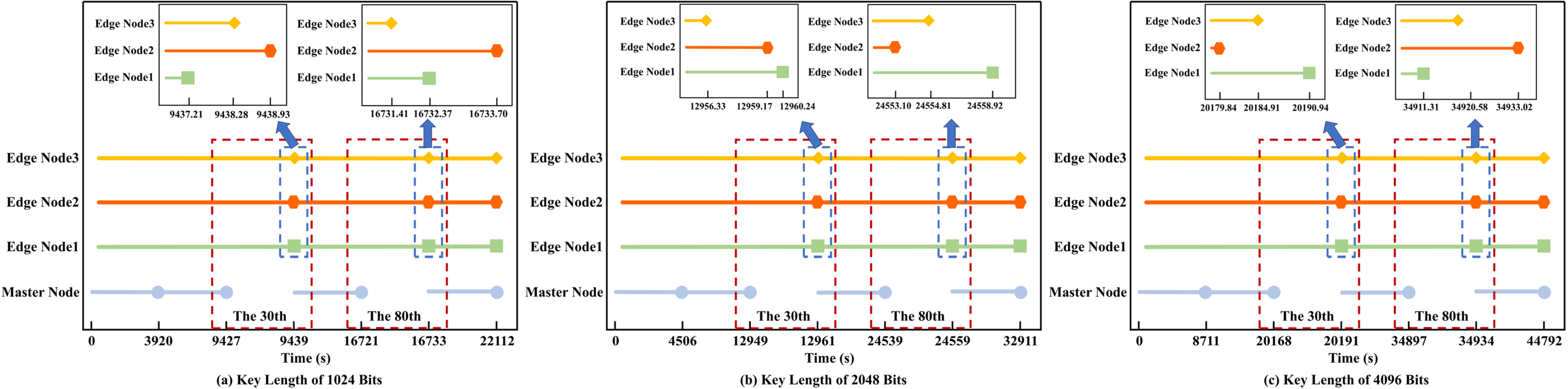}
\captionsetup{format=plain,justification=centering}
\caption{The computational latency among nodes of proposed GPU-Accelerated 3P-ADMM-PC2.}
\label{com-delay}
\end{figure*}

To further analyze the acceleration performance and feasibility of the proposed 3P-ADMM-PC2 framework, we evaluated the latency characteristics of both GPU and CPU implementations under larger-scale sparse problems and a greater number of edge nodes. When the dimension of matrix $\mathbf{A}$ is increased to $\mathbf{A} \in \mathbb{R}^{10000 \times 65536}$ with 10 edge nodes, the latency per iteration for both GPU and CPU computations under different key lengths is shown in Table \ref{10-edge-nodes-CPU-GPU-1024}, \ref{10-edge-nodes-CPU-GPU-2048} and \ref{10-edge-nodes-CPU-GPU-4096}.

\begin{table*}[htbp]
\centering
\caption{Latency comparison of encryption and decryption using GPU vs. CPU with 10 edge nodes and 1024-bit key length.}
\label{10-edge-nodes-CPU-GPU-1024}
\begin{tabular}{|c|c|c|c|c|c|c|c|c|}
\hline
      \textbf{Hardware Environment} & \multicolumn{4}{|c|}{\textbf{GPU}} & \multicolumn{4}{|c|}{\textbf{CPU}} \\
\cmidrule(lr){1-1} \cmidrule(lr){2-5} \cmidrule(lr){6-9}
       \multirow{2}{*}{\textbf{Computation Phase}} & \multirow{2}{*}{\textbf{Initialization}} & \multicolumn{3}{|c|}{\textbf{Iterative Computation}} & \multirow{2}{*}{\textbf{Initialization}} & \multicolumn{3}{|c|}{\textbf{Iterative Computation}} \\
\cmidrule(lr){3-5} \cmidrule(lr){7-9}
      &  & 30th & 80th & 100th &   & 30th & 80th & 100th \\
\hline
      \textbf{Computation Completion Time} (s) & 1735 & 4679 & 9596 & \textbf{11697} & 4112 & 11818 & 24693 & \textbf{29858}\\

      \textbf{Computational Latency} (s) & 1411 & 93 & 98 & 94 & 3771 & 253 & 266 & 260 \\

      \textbf{Communication Latency} (s) & 294 & 8 & 7 & - & 289 & 8 & 8 & -\\

      \textbf{Waiting Latency} (s) & 24 & 6 & 6 & - & 44 & 8 & 9 & -\\
\hline
\end{tabular}
\end{table*}

As shown in Table \ref{10-edge-nodes-CPU-GPU-1024}, when the key length is 1024 bits, the total algorithm runtime using GPU-based encryption and decryption is 11697 seconds, while that using CPU-based encryption and decryption is 29858 seconds. When the number of edge nodes is fixed, the dimensionality and data type of the transmitted parameters remain consistent. Therefore, the communication overhead differs only marginally between GPU and CPU implementations, with variations primarily arising from real-time network fluctuations. Due to the superior computational power of the master node compared to the edge nodes, the master completes its current computational task earlier. This results in a waiting latency of approximately 5 seconds when using GPU encryption and decryption, and about 7 seconds when using CPU encryption and decryption. Additionally, due to uneven plaintext lengths and varying real-time hardware loads, the completion time among edge nodes differs by approximately 1 second with GPU encryption and decryption, and about 2 seconds with CPU encryption and decryption. Overall, the computational overhead per iteration using GPU-based encryption and decryption is significantly lower than that of CPU-based encryption and decryption.

\begin{table*}[htbp]
\centering
\caption{Latency comparison of encryption and decryption using GPU vs. CPU with 10 edge nodes and 2048-bit key length.}
\label{10-edge-nodes-CPU-GPU-2048}
\begin{tabular}{|c|c|c|c|c|c|c|c|c|}
\hline
      \textbf{Hardware Environment} & \multicolumn{4}{|c|}{\textbf{GPU}} & \multicolumn{4}{|c|}{\textbf{CPU}} \\
\cmidrule(lr){1-1} \cmidrule(lr){2-5} \cmidrule(lr){6-9}
       \multirow{2}{*}{\textbf{Computation Phase}} & \multirow{2}{*}{\textbf{Initialization}} & \multicolumn{3}{|c|}{\textbf{Iterative Computation}} & \multirow{2}{*}{\textbf{Initialization}} & \multicolumn{3}{|c|}{\textbf{Iterative Computation}} \\
\cmidrule(lr){3-5} \cmidrule(lr){7-9}
      &  & 30th & 80th & 100th &   & 30th & 80th & 100th \\
\hline
      \textbf{Computation Completion Time} (s) & 4007 & 9107 & 17811 & \textbf{20537} & 6828 & 17141 & 33424 & \textbf{41016} \\

      \textbf{Computational Latency} (s) & 3669 & 159 & 164 & 161 & 6462 & 351 & 364 & 353\\

      \textbf{Communication Latency} (s) & 292 & 11 & 10 & - & 297 & 10 & 12 & - \\

      \textbf{Waiting Latency} (s) & 41 & 3 & 5 & - & 62 & 13 & 16 & - \\
\hline
\end{tabular}
\end{table*}

As shown in Table \ref{10-edge-nodes-CPU-GPU-2048}, when the key length is 2048 bits, the computational overhead per iteration remains lower with GPU-based encryption and decryption than with CPU-based encryption and decryption.    Interestingly, although the master node possesses superior computational power compared to the edge nodes, it faces growing challenges as the key length and task dimensionality increase. In each iteration, the edge nodes now complete their tasks earlier than the master node. When using GPU encryption and decryption, the edge nodes waits approximately 4 seconds for the master node, while the edge nodes wait about 2 seconds for each other. In contrast, with CPU encryption and decryption, the edge nodes experiences a waiting latency of around 9 seconds, and the edge nodes wait approximately 4 seconds for each other.

\begin{table*}[htbp]
\centering
\caption{Latency comparison of encryption and decryption using GPU vs. CPU with 10 edge nodes and 4096-bit key length.}
\label{10-edge-nodes-CPU-GPU-4096}
\begin{tabular}{|c|c|c|c|c|c|c|c|c|}
\hline
      \textbf{Hardware Environment} & \multicolumn{4}{|c|}{\textbf{GPU}} & \multicolumn{4}{|c|}{\textbf{CPU}} \\
\cmidrule(lr){1-1} \cmidrule(lr){2-5} \cmidrule(lr){6-9}
       \multirow{2}{*}{\textbf{Computation Phase}} & \multirow{2}{*}{\textbf{Initialization}} & \multicolumn{3}{|c|}{\textbf{Iterative Computation}} & \multirow{2}{*}{\textbf{Initialization}} & \multicolumn{3}{|c|}{\textbf{Iterative Computation}} \\
\cmidrule(lr){3-5} \cmidrule(lr){7-9}
      &  & 30th & 80th & 100th &   & 30th & 80th & 100th \\
\hline
      \textbf{Computation Completion Time} (s) & 6476 & 15021 & 29203 & \textbf{34967} & 13719 & 32904 & 66595 & \textbf{79338}\\

      \textbf{Computational Latency} (s) & 6114 & 278 & 288 & 281 & 13314 & 644 & 681 & 659 \\

      \textbf{Communication Latency} (s) & 295 & 13 & 11 & - & 297 & 12 & 11 & -\\

      \textbf{Waiting Latency} (s) & 56 & 12 & 10 & - & 95 & 20 & 17 & -\\
\hline
\end{tabular}
\end{table*}

As shown in Table \ref{10-edge-nodes-CPU-GPU-4096}, when the key length is 4096 bits, the computational overhead for encryption and decryption increases significantly. However, the computational cost per iteration remains lower with GPU-based encryption and decryption than with CPU-based encryption and decryption. It is noteworthy that the waiting latency among nodes increases with the key length. This is because the fluctuation in computational effort resulting from the non-uniform distribution of plaintext becomes more pronounced as the key length grows.

When the number of edge nodes is 3 and the dimension of $\mathbf{A}$ is $\mathbf{A} \in \mathbb{R}^{10000 \times 65536}$, no scenario occurs where edge nodes wait for the master node, regardless of whether the key length is 1024, 2048, or 4096 bits. Therefore, in practical computational tasks, the balance among task scale, number of edge nodes, and disparities in hardware computational power directly affects both computational and waiting latency in each specific operation. When the number of edge nodes is large, the parallel computation across multiple edge nodes can even surpass the master node in speed. In future work, we will continue to investigate the constraints among task scale, number of edge nodes, and hardware capability gaps, and further optimize task allocation strategies to reduce overhead.
\vspace{-10pt}

\subsection{Application Example: Power Network Reconstruction}
In the environment depicted in Fig. \ref{hardware-environment}, consider a distributed edge privacy-computing network comprising one master node and three edge nodes. The goal is to recover the topology of a power network with $N$ buses from $M$ sets of observed data. For the adjacency matrix $\mathbf{d}_i$ of bus $i$, have the following optimization problem:
\begin{equation}\label{power_network_rec}
\arg\min_{\mathbf{d}_i\in\mathbb{R}^N}\frac{1}{2}\left\| \mathbf{S}_i-\mathbf{\Phi}_i\mathbf{d}_i  \right\|_{2}^{2}+\lambda {{\left\| \mathbf{d}_i  \right\|}_{1}},
\end{equation}
where $\mathbf{S}_i \in \mathbb{R}^M$ represents the combination of current values at different time points for bus $i$; $\mathbf{\Phi}_i \in \mathbb{R}^{M \times N} (M \ll N)$ denotes the observation matrix of voltage differences and resistance values at different time points for bus $i$; and $\mathbf{d}_i \in \mathbb{R}^N$ is the adjacency vector for bus $i$.
In a LAN, using the MATPOWER toolbox \footnote{https://matpower.org/}, a large scale real power network with 13569 buses is utilized, employing the local reconstruction method mentioned in \cite{liu2023distributed}. For the high dimensionality of network, the centralized methods for network reconstruction is no longer feasible. The accuracy of the reconstruction is evaluated using the binary classification metrics area under the receiver operating characteristic (AUROC) and area under the precision-recall curve (AUPRC). In binary classification metrics, False Positive Rate and Recall are used as the x-axis, while True Positive Rate and Precision are used as the y-axis. The area under the curve represents the values of AUROC and AUPRC, respectively.

Fig. \ref{AUROC_AUPRC} displays the accuracy of network reconstruction for the Dis.-ADMM scheme and the proposed GPU-accelerated 3P-ADMM-PC2 scheme under different data ratios ($R_D$). The proposed GPU-accelerated 3P-ADMM-PC2 scheme demonstrates high reconstruction accuracy. The AUROC and AUPRC of the Dis.-ADMM scheme coincide with those of the proposed GPU-accelerated 3P-ADMM-PC2 scheme, indicating once again that the accuracy loss of the proposed quantization method is negligible.
\vspace{-10pt}
\begin{figure}[hbtp]
\centering
\includegraphics[width=3.5in]{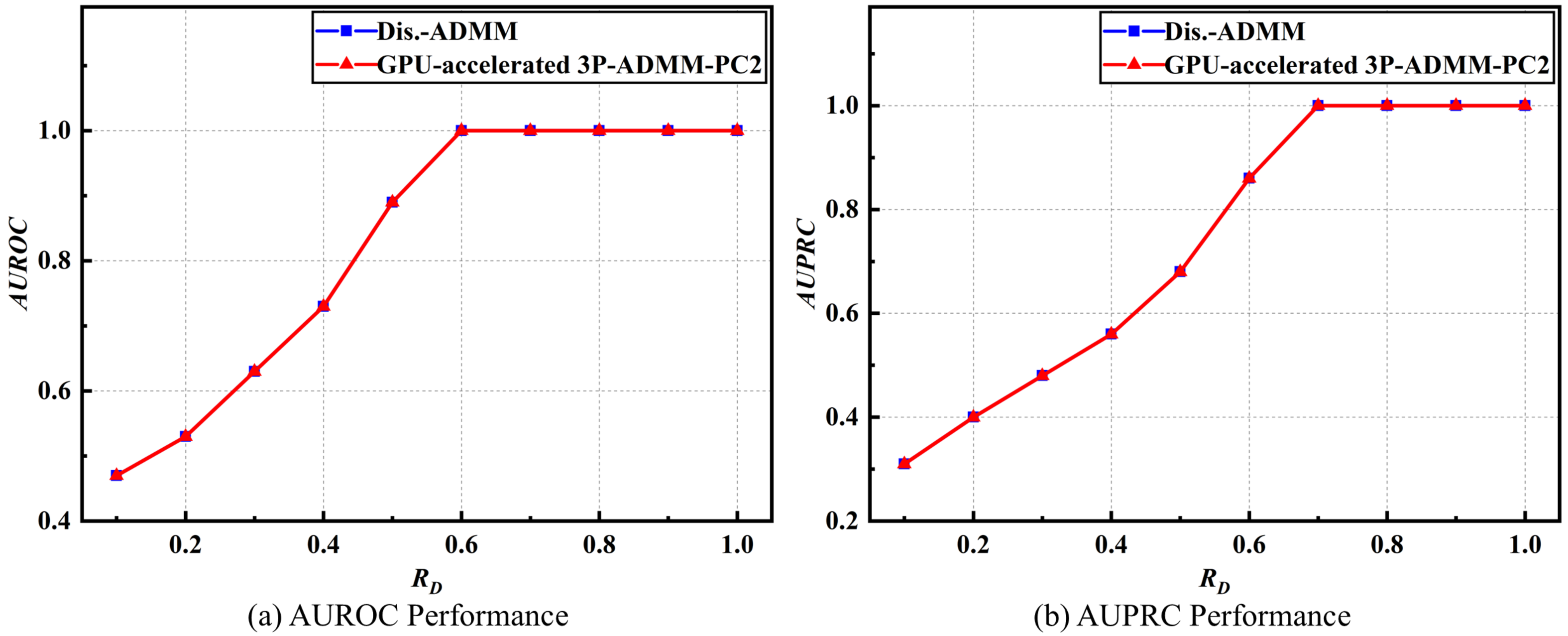}
\captionsetup{format=plain,justification=centering}
\caption{The power grid reconstruction results of 13659 Buses, where $\Delta={10}^{15}$ and the key length is 2048 bits of GPU-accelerated 3P-ADMM-PC2.}
\label{AUROC_AUPRC}
\end{figure}
\section{Conclusion}
This paper first proposed a parallel collaborative ADMM privacy computing algorithm, named 3P-ADMM-PC2, for distributed edge networks, where the edge nodes and master node complete the ADMM task calculation in three phases
(3P), obtaining parallel collaborative privacy computing (PC2). Specially, the quantization method of proposed 3P-ADMM-PC2 addresses the limitations on plaintext for homomorphic encryption without affecting its homomorphic operations. Additionally, the ModExp of large integers has been successfully deployed on GPU across different platforms, enabling parallel encryption, decryption, and homomorphic computations, and facilitating multi-platform collaborative computing. To reduce the computational burden on the master node, a collaborative encryption and decryption algorithm between the master node and edge nodes, named GPU-accelerated 3P-ADMM-PC2, is finally proposed.

\bibliographystyle{IEEEtran}
\bibliography{trans}
\vspace{0.9em}

\begin{IEEEbiography}[{\includegraphics[width=1in,height=1.25in,clip,keepaspectratio]{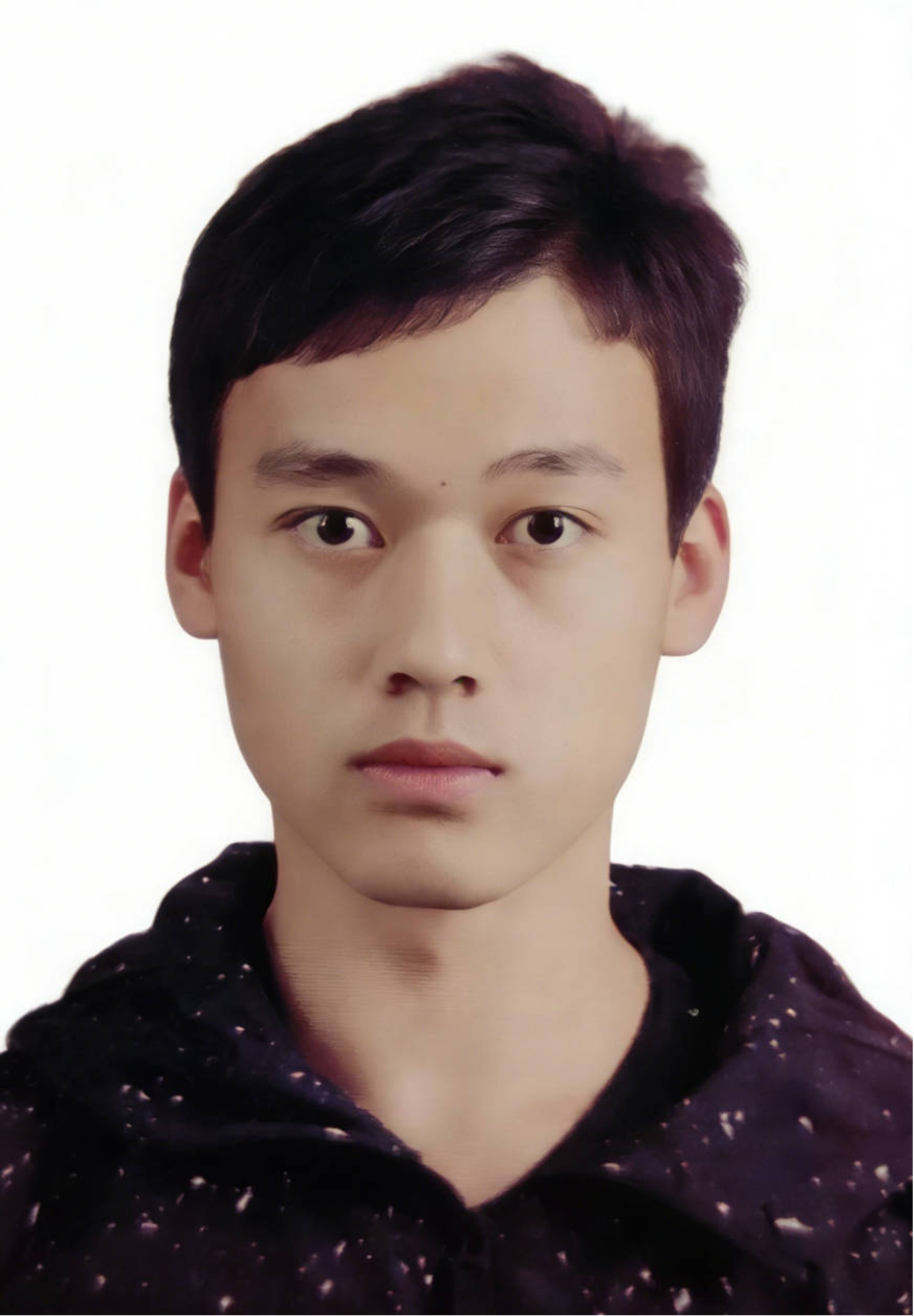}}]
{Mengchun Xia} is currently pursuing a master's degree at Xizang University in Lhasa, China. His current research interests include privacy computing and distributed computing.
\end{IEEEbiography}

\begin{IEEEbiography}[{\includegraphics[width=1in,height=1.25in,clip,keepaspectratio]{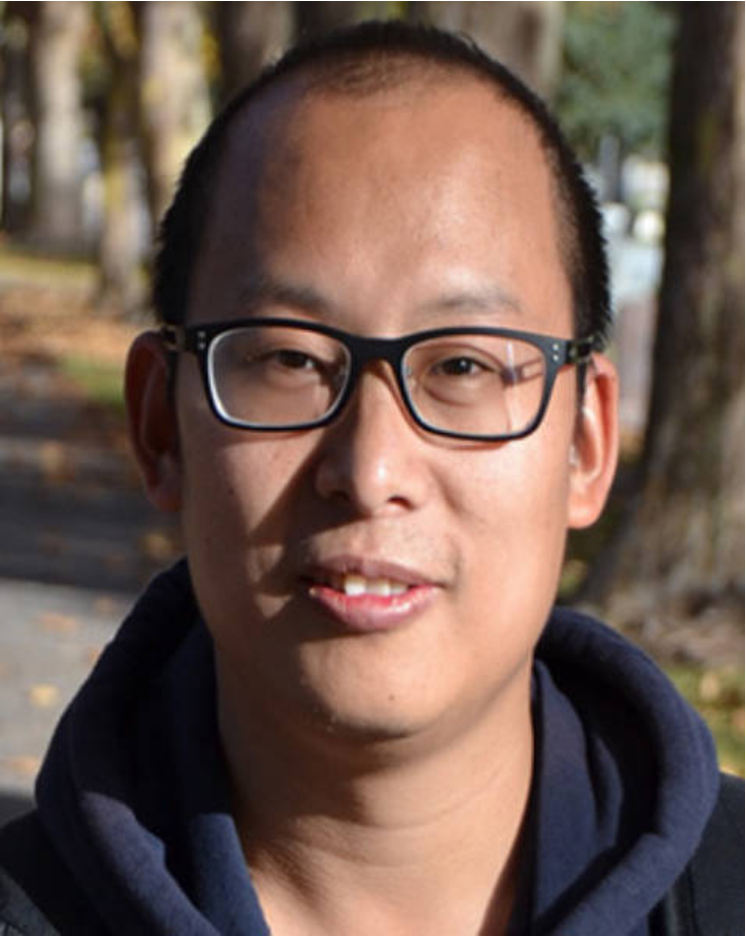}}]
{Zhicheng Dong} (M'14) received the B.E., M.S. and Ph.D. degrees from the School of information Science and Technology from Southwest Jiaotong University, Chengdu, China, in 2004, 2008 and 2016, respectively. From 2013 to 2014, he was a
Visiting Scholar with the Department of Electrical and Computer Engineering at Utah State University, Logan, UT, USA. From 2018 to 2019, he was a Research Fellow with the Department of Electrical and Computer Engineering, Columbia University, New York, NY, USA. Since 2018, he has been a Professor with the School of Information Science and Technology, Xizang University, Lhasa, China. He was the recipient of the EAI GameNets2022 Best Paper Award. He was a Track Chair and Technical Sessions Chair of 2022 IEEE ICICN. He was a TPC members of the major international conferences, such as IEEE ICC, IEEE GLOBECOM, IEEE WCNC, IEEE VTC, IEEE PIMRC. He was a Reviewer for many well known journals, such as IEEE Wireless Communications Magazine, IEEE JOURNAL ON SELECTED AREAS IN COMMUNICATIONS, IEEE TRANSACTIONS ON COMMUNICATIONS, IEEE TRANSACTIONS ON VEHICULAR TECHNOLOGY. His research interests include artificial intelligence, computer vision, adaptation technology, performance analysis, and signal processing.
\end{IEEEbiography}

\begin{IEEEbiography}[{\includegraphics[width=1in,height=1.25in,clip,keepaspectratio]{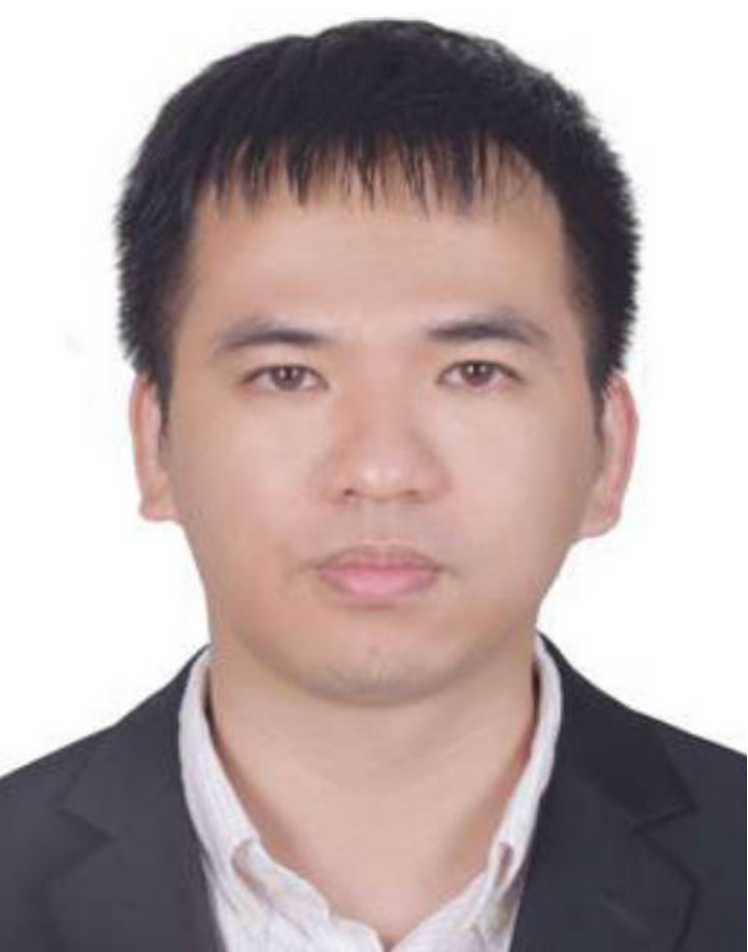}}]
{Donghong Cai} (SM'25, M'20) received the B.S. degree from the School of Mathematics and Information Sciences, Shaoguan University, Shaoguan, China, in 2012, and the M.S. and Ph.D. degrees from Southwest Jiaotong University, Chengdu, China, in 2015 and 2020, respectively. From October 2017 to October 2018, he was a visiting Ph.D. student with Lancaster University, Lancaster, U.K., and the University of Manchester, Manchester, U.K. He served as a Guest Editor for Physical Communication and Journal on Wireless Communications and Networking. He currently serves as associate editor for IEEE Open Journal of Signal Processing and International Journal of Communication Systems. He is currently an associate professor with the college of cyber security, Jinan University, Guangzhou, China. His current research interests include signal detection, security coding, distributed Internet of Things, privacy preservation, machine learning, and nonorthogonal multiple access.
\end{IEEEbiography}

\begin{IEEEbiography}[{\includegraphics[width=1in,height=1.25in,clip,keepaspectratio]{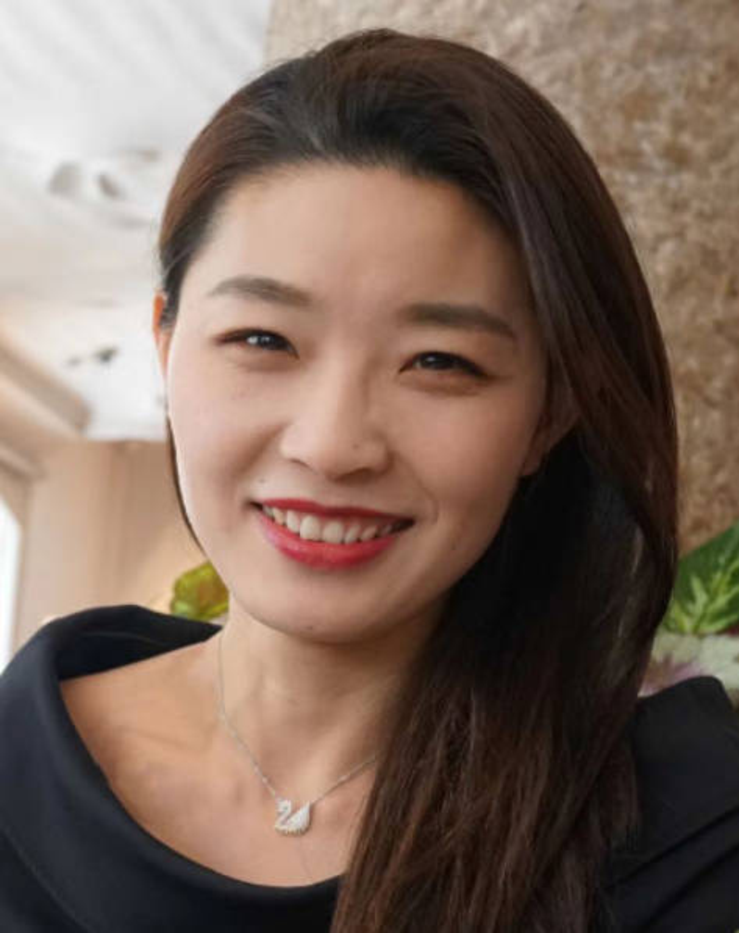}}]
{Fang Fang} (M'18-SM'23) is currently an Assistant Professor in the Department of Electrical and Computer Engineering and the Department of Computer Science, Western University, Canada. She received the Ph.D. degree in electrical engineering from the University of British Columbia (UBC), Canada. Her current research interests include machine learning for intelligent wireless communications, non-orthogonal multiple access (NOMA), reconfigurable intelligent surface (RIS), multi-access edge computing (MEC), Semantic Communications and Edge AI, etc. Dr. Fang has been serving as a general chair for EAI GameNets 2022, a publications chair for IEEE VTC 2023 Fall, and a symposium chair for IEEE Globecom 2023, IEEE VTC 2024 Spring and IEEE ICC 2025. Dr. Fang currently serves as an Editor for IEEE Communication Letters and IEEE Open Journal of the Communications Society (OJ-COM). She received the Exemplary Reviewer Certificates from the IEEE Transactions on Communications in 2017 and 2021, as well as the Exemplary Editor Certificates from IEEE OJ-COM in 2021 and 2023. Dr. Fang won the IEEE SPCC Early Achievement Award in 2023.
\end{IEEEbiography}

\begin{IEEEbiography}[{\includegraphics[width=1in,height=1.25in,clip,keepaspectratio]{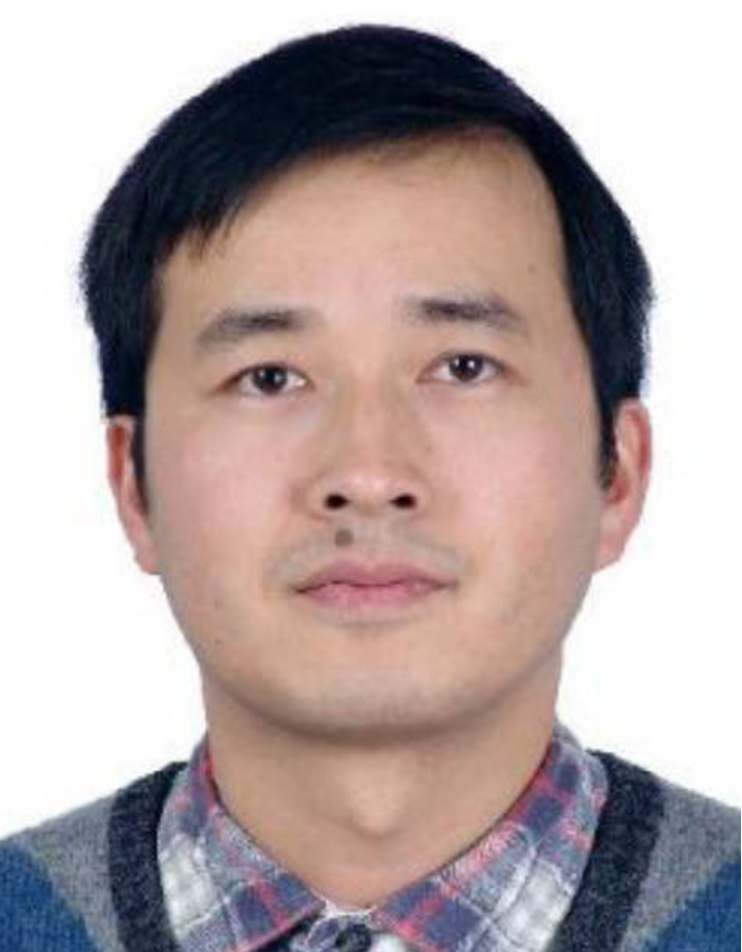}}]
{Lisheng Fan} received the bachelor's degree from the Department of Electronic Engineering, Fudan University, in 2002, the master's degree from the Department of Electronic Engineering, Tsinghua University, China, in 2005, and the Ph.D. degree from the Department of Communications and Integrated Systems, Tokyo Institute of Technology, Japan, in 2008. He is currently a Professor with Guangzhou University. He has published many articles in international journals, such as IEEE TRANSACTIONS ON WIRELESS COMMUNICATIONS, IEEE TRANSACTIONS ON COMMUNICATIONS, and IEEE TRANSACTIONS ON INFORMATION THEORY, and papers in conferences, such as IEEE ICC, IEEE Globecom, and IEEE WCNC. His research interests include wireless cooperative communications, physical-layer secure communications, interference modeling, and system performance evaluation. He is a Guest Editor of EURASIP Journal on Wireless Communications and Networking. He served as the Chair for Wireless Communications and Networking Symposium for Chinacom 2014. He has also served as a member of Technical Program Committees for IEEE conferences, such as Globecom, ICC, WCNC, and VTC.
\end{IEEEbiography}

\begin{IEEEbiography}[{\includegraphics[width=1in,height=1.25in,clip,keepaspectratio]{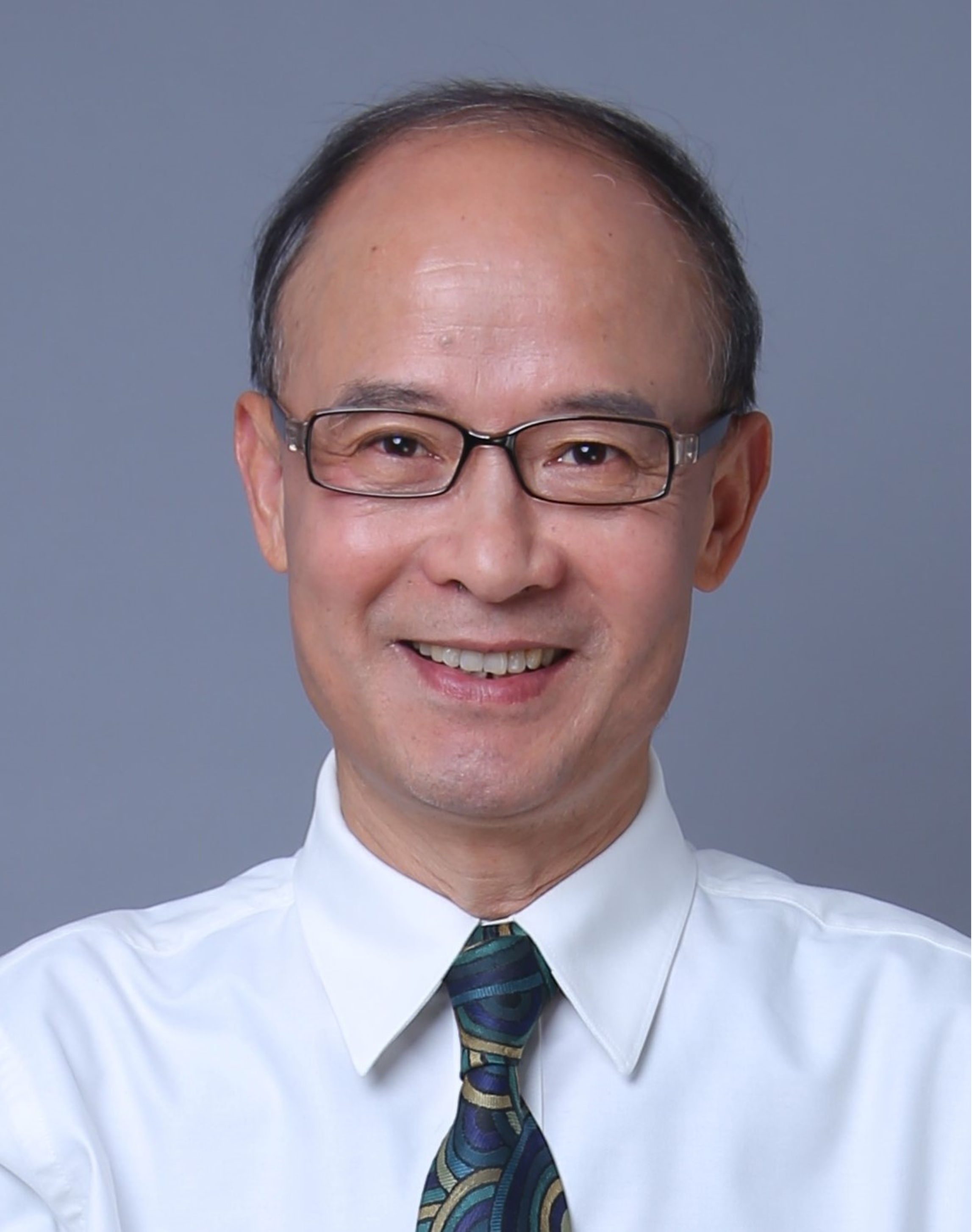}}]
{Pingzhi Fan} (Fellow, IEEE) received the M.Sc. degree in computer science from Southwest Jiaotong University, China, in 1987, and the Ph.D. degree in electronic engineering from Hull University, U.K., in 1994. He is currently a distinguished professor of Southwest Jiaotong University (SWJTU), honorary dean of the SWJTU-Leeds Joint School (2015-), honorary professor of the University of Nottingham (Ningbo, 2025), and a visiting professor of Leeds University, UK (1997-). He is a recipient of the UK ORS Award (1992), the National Science Fund for Distinguished Young Scholars (1998, NSFC), IEEE VT Society Jack Neubauer Memorial Award (2018), IEEE SP Society SPL Best Paper Award (2018), IEEE VT Society Best Magazine Paper Award (2023), and several IEEE conference best paper awards. He served as chief scientist of a National 973 Plan Project (MoST, 2012.1-2016.12). He also served as general chair or TPC chair of a number of IEEE conferences, including VTC2016Spring, ITW2018, IWSDA2022, PIMRC2023, VTC2025Fall, as well as the coming ISIT2026, ICC2028. His research interests include high mobility wireless communications, multiple access techniques, ISAC, signal design \& coding, etc. He is an IEEE VTS Distinguished Speaker (2022-2028), a fellow of IEEE, IET, CIE and CIC.
\end{IEEEbiography}

\end{document}